\documentclass[final,sort&compress,3p]{elsarticle}
\usepackage{amsfonts}
\usepackage{amsthm}
\usepackage{amsmath}

\newtheorem{theorem}{Theorem}[section]
\newtheorem{corollary}{Corollary}[section]

\newcommand{\Eqref}[1]{Equation~\eqref{#1}}
\newcommand{\Eqsref}[1]{Equations~\eqref{#1}}
\newcommand{\secref}[1]{Section~\ref{#1}}

\newcommand{\Secref}[1]{Section~\ref{#1}}

\newcommand{\subsecref}[1]{Subsection~\ref{#1}}

\newcommand{\thmref}[1]{Theorem~\ref{#1}}

\newcommand{\braket}[1]{\langle{#1}\rangle}

\newcommand{\lshad}{[\![}
\newcommand{\rshad}{]\!]}
\newcommand{\sdot}{\,\cdot\,}
\newcommand{\dd}{\mathrm{d}}
\newcommand{\vd}[1]{\dd{#1}\,}
\newcommand{\ud}{\,\dd}
\newcommand{\nbr}[1]{$#1$\nobreakdash-\hspace{0pt}}
\providecommand{\abs}[1]{\lvert#1\rvert}
\providecommand{\Abs}[1]{\left\lvert#1\right\rvert}

\DeclareMathOperator{\tr}{tr}
\DeclareMathOperator{\sgn}{sgn}

\journal{arXiv}

\begin{document}

\begin{frontmatter}

\title{Canonical transformations in quantum mechanics}
\author{Maciej B{\l}aszak}
\ead{blaszakm@amu.edu.pl}
\author{Ziemowit Doma{\'n}ski}
\ead{ziemowit@amu.edu.pl}
\address{Faculty of Physics, Adam Mickiewicz University\\
         Umultowska 85, 61-614 Pozna{\'n}, Poland}

\begin{abstract}
This paper presents the general theory of canonical transformations of
coordinates in quantum mechanics. First, the theory is developed in the
formalism of phase space quantum mechanics. It is shown that by transforming a
star-product, when passing to a new coordinate system, observables and states
transform as in classical mechanics, i.e., by composing them with a
transformation of coordinates. Then the developed formalism of coordinate
transformations is transferred to a standard formulation of quantum mechanics.
In addition, the developed theory is illustrated on examples of particular
classes of quantum canonical transformations.
\end{abstract}


\end{frontmatter}

\section{Introduction}
\label{sec:1}
In classical Hamiltonian mechanics transformations of phase space coordinates
(especially canonical transformations of coordinates) are an important part of
the theory. Since quantum mechanics arises from Hamiltonian mechanics it is
natural to ask if it is possible to introduce the concept of canonical
transformations of coordinates in quantum mechanics. The importance of this
problem was evident to scientists from the early days of quantum theory. The
development of the theory of canonical transformations of coordinates in quantum
mechanics is mainly contributed to Jordan, London and Dirac back in 1925
\cite{Born:1926,Jordan:1926a,Jordan:1926b,London:1926,London:1927,Dirac:1927,%
Dirac:1958} and it is still an area of intense research.

In the usual approach to canonical transformations in quantum mechanics one
identifies canonical transformations with unitary operators defined on a Hilbert
space. Such approach was used by Mario Moshinsky and his collaborators in a
series of papers \cite{Mello:1975,Moshinsky:1978,Moshinsky:1979,%
Garcia-Calderon:1980,Dirl:1988}. Also other researchers used such approach
\cite{Lee.Yi:1995,Kim:1999}. Worth noting are also papers of
\citet{Anderson:1993,Anderson:1994} where an extension of canonical
transformations to non-unitary operators is presented. Nevertheless, after so
many years of efforts, there is still lack of a general theory of coordinate
transformations in quantum mechanics, including a satisfactory complete theory
of canonical transformations. With the following paper we are trying to fill
this gap, at least to some extend.

A unitary operator $\hat{U}$ associated with some canonical change
of coordinates transforms vector states from one coordinate system
to the other: $\phi'=\hat{U}\phi$. Also observables, being
operators on a Hilbert space of states, transform according to a
prescription
\begin{equation}
\hat{A}' = \hat{U} \hat{A} \hat{U}^{-1},
\label{eq:1.1}
\end{equation}
where $\hat{A}$ is an initial operator and $\hat{A}'$ an operator
after the change of coordinates. In particular, operators
$\hat{q}$, $\hat{p}$ of position and momentum transform according
to \eqref{eq:1.1}. Observe that the commutator of operators of
position and momentum as well as their Hermitian property do not
change after a unitary transformation. For this reason unitary
operators are identified with canonical transformations.

To some unitary operators one can try to associate transformations
defined on a classical phase space. It can be done using a Born's
quantization rule. This rule formulated by Max Born states that to
every classical observable $A(x,p)$ (real-valued function defined
on a phase space) written in a Cartesian coordinate system one can
associate an operator on a Hilbert space by substituting for $x$
and $p$ operators of position and momentum $\hat{q}$, $\hat{p}$
and symmetrically ordering them. If one now considers a
transformation of phase space coordinates (possibly depending on
an evolution parameter $t$)
\begin{equation*}
T = (Q, P) \colon \mathbb{R}^2 \supset U \to W \subset \mathbb{R}^2,
\end{equation*}
then one can define operators
\begin{align*}
\hat{Q} & = Q(\hat{q},\hat{p},t), \\
\hat{P} & = P(\hat{q},\hat{p},t).
\end{align*}
If $[\hat{Q},\hat{P}] = i\hbar$ then the operators $\hat{Q}$, $\hat{P}$ can be
thought of as operators of position and momentum associated to a new coordinate
system, and $T$ is then called a quantum canonical transformation. In some cases
the operator $\hat{Q}$ can be written in a position representation, i.e., as an
operator of multiplication by a coordinate variable. In fact, there exist a
measure $\mu$ and a unitary operator $\hat{U}_T \colon \mathcal{H} \to
L^2(\sigma(\hat{Q}),\mu)$ such that
\begin{equation*}
\hat{U}_T \hat{Q} \hat{U}_T^{-1} = x',
\end{equation*}
where $\sigma(\hat{Q})$ is the spectrum of the operator $\hat{Q}$. The
operator $\hat{U}_T$ is precisely the unitary operator corresponding to the
quantum canonical transformation $T$.

At this point we arrive to a problem concerning the Born's quantization
procedure. Namely, we could try to quantize a classical Hamiltonian system in
some arbitrary canonical coordinate system using the Born's quantization rule.
However, this way we would end up in non-equivalent quantum systems. As an
example let us consider a classical harmonic oscillator which time evolution is
governed by the Hamiltonian
\begin{equation*}
H(x,p) = \frac{1}{2} \left( p^2 + \omega^2 x^2 \right).
\end{equation*}
In accordance with the Born's quantization rule to $H$ corresponds the following
operator
\begin{equation*}
\hat{H} = H(\hat{q},\hat{p})
= \frac{1}{2} \left( \hat{p}^2 + \omega^2 \hat{q}^2 \right).
\end{equation*}
Performing a classical canonical transformation of coordinates
$T \colon (\mathbb{R} \setminus \{0\}) \times \mathbb{R} \to
(\mathbb{R} \setminus \{0\}) \times \mathbb{R}$, $T(x',p') = (x,p)$ where
\begin{equation*}
x = \begin{cases}
        \phantom{-} \sqrt{\abs{2x'}}, & x' > 0 \\
        -\sqrt{\abs{2x'}}, & x' < 0
    \end{cases}, \quad
p = p'\sqrt{\abs{2x'}},
\end{equation*}
with the inverse $T^{-1}$ in the form
\begin{equation*}
x' = \begin{cases}
         \phantom{-} \frac{1}{2}x^2, & x > 0 \\
         -\frac{1}{2}x^2, & x < 0
     \end{cases}, \quad
p' = \abs{x}^{-1} p,
\end{equation*}
the Hamiltonian $H$ transforms to the following function
\begin{equation}
H'(x',p') = H(T(x',p')) = \abs{x'}p'^2 + \omega^2 \abs{x'}.
\label{eq:1.2}
\end{equation}
Applying the Born's quantization rule to $H'$ with symmetric ordering of
canonical operators of position and momentum, $\hat{q}' = x'$,
$\hat{p}' = -i\hbar\partial_{x'}$, results in the following operator
\begin{equation*}
\hat{H}' = H'(\hat{q}',\hat{p}') = \frac{1}{2} \abs{\hat{q}'} \hat{p}'^2
    + \frac{1}{2} \hat{p}'^2 \abs{\hat{q}'} + \omega^2 \abs{\hat{q}'}.
\end{equation*}
Hamiltonians $\hat{H}$ and $\hat{H}'$ are not unitarily equivalent
--- they describe different quantum systems. As we will show in
this paper, the apparent inconsistency of the quantization can be
solved by changing the ordering of operators $\hat{q}'$,
$\hat{p}'$ in the transformed Hamiltonian $H'$ in an appropriate
way, which will be equivalent to the existence of a respective
unitary operator related to the transformation.

In order to control the whole procedure we take a different route
for developing a theory of transformations of coordinates in
quantum mechanics, to that usually found in the literature.
Namely, first we consider a quantum theory on a phase space, and
then pass to the standard description of quantum mechanics. Such
approach seems to be more natural as the concept of a change of
coordinates can be introduced in a similar manner as in classical
Hamiltonian theory.

In the usual approach to coordinate transformations in phase space quantum
mechanics \cite{Hietarinta:1982,Dereli:2009,Curtright:1998} one modifies an action of a
transformation of coordinates on observables and states, from the simple
composition of maps. This is done to receive a consistent theory. We present
a slightly different, although equivalent, approach in which we transform not
only observables and states but also a star-product. Such an idea the reader can
also find in \cite{Dias:2004}. Then the action of a coordinate transformation on
phase space functions can remain to be a simple composition of maps. Moreover,
we present a systematic way of passing the developed theory of transformations
to the standard description of quantum mechanics. Such approach of dealing with
change of coordinates in quantum theory proves to be very effective and gives
means to develop a general theory of transformations of coordinates in quantum
mechanics which is still missing after almost 90 years of its development.

The paper is organized as follows. In \secref{sec:2} we review basic concepts of
phase space quantum mechanics focusing on properties important in developing a
transformation theory. \Secref{sec:3} presents the general theory of canonical
transformations of coordinates in phase space quantum mechanics. \Secref{sec:4}
contains the passage of the quantum transformation theory developed on the phase
space to the ordinary quantum mechanics, using the results from previous
section. \Secref{sec:5} presents couple classes of transformations of
coordinates illustrating the developed formalism. In \secref{sec:6} we give
final remarks and conclusions.

\section{Phase space quantum mechanics}
\label{sec:2}

\subsection{Preliminaries}
\label{subsec:2.1}
There exist many approaches to quantum mechanics. The one which seems to be the
most suited to develop a theory of transformations of phase space coordinates is
a phase space quantum mechanics also referred to as a deformation quantization
(see \cite{Bayen:1978a,Bayen:1978b}, and \cite{Dito.Sternheimer:2002,Gutt:2000,%
Weinstein:1994b} for recent reviews). The review of phase space quantum
mechanics which we will present is based on our earlier work
\cite{Blaszak:2012}.

The idea behind the phase space quantum mechanics relies on an appropriate
deformation of a classical Hamiltonian system. Such approach to quantum
theory uses a mathematical language very similar to that of a classical
theory. This have many benefits. On the one hand it allows a straightforward
introduction of concepts from classical mechanics into the quantum counterpart
(for example, the concept of a transformation of phase space coordinates), on
the other hand it gives in a natural way a classical Hamiltonian mechanics as a
limiting case, hence allowing investigation of quantum systems in a classical
limit.

The procedure of a deformation should be performed with respect to some
parameter $\hbar$ (which is taken to be a Planck's constant) so that in the
limit $\hbar \to 0$ the quantum theory should reduce to the classical theory.

The deformation of a classical Hamiltonian system can be fully
given by deforming an algebraic structure of a classical Poisson
algebra. This will then yield a deformation of a phase space (a
Poisson manifold) to a noncommutative phase space (a
noncommutative Poisson manifold), a deformation of classical
states to quantum states and a deformation of classical
observables to quantum observables.

First, let us deal with a deformation of a phase space. A Poisson
manifold $(M,\mathcal{P})$ ($\mathcal{P}$ being a Poisson tensor)
is fully described by a Poisson algebra $\mathcal{A}_C =
(C^\infty(M),\cdot,\{\sdot,\sdot\})$ of smooth complex-valued
functions on the phase space $M$, where $\cdot$ is a point-wise
product of functions and $\{\sdot,\sdot\}$ is a Poisson bracket
induced by a Poisson tensor $\mathcal{P}$. Hence by deforming
$\mathcal{A}_C$ to some noncommutative algebra $\mathcal{A}_Q =
(C^\infty(M),\star,\lshad\sdot,\sdot\rshad)$, where $\star$ is
some noncommutative associative product of functions being a
deformation of a point-wise product, and $\lshad\sdot,\sdot\rshad$
is a Lie bracket satisfying the Leibniz's rule and being a
deformation of the Poisson bracket $\{\sdot,\sdot\}$, we can think
of a quantum Poisson algebra $\mathcal{A}_Q$ as describing a
noncommutative Poisson manifold.

The algebra $\mathcal{A}_C$ contains in particular a subset of
classical observables, whereas $\mathcal{A}_Q$ contains a subset
of quantum observables. Note that quantum observables are
functions on the phase space $M$ similarly as in classical
mechanics. Furthermore, classical observables are real-valued
functions from $\mathcal{A}_C$, i.e., self-adjoint functions with
respect to the complex-conjugation --- an involution in the
algebra $\mathcal{A}_C$. Quantum observables should also be
self-adjoint functions with respect to an involution in the
algebra $\mathcal{A}_Q$. However, in general the
complex-conjugation do not need to be an involution in
$\mathcal{A}_Q$. Thus in $\mathcal{A}_Q$ we have to introduce some
involution which would be a deformation of the complex-conjugation
\cite{Blaszak:2012}. As a consequence, quantum observables
(self-adjoint functions with respect to the quantum involution)
might be complex and $\hbar$-dependent.

There is a vast number of equivalent quantization schemes (see
\cite{Blaszak:2012} for review of the subject) which yield a
quantization equivalent to a standard approach to quantum
mechanics but giving different orderings of position and momentum
operators. From this diversity of quantization schemes the
simplest one is a Moyal quantization. It follows from the fact
that for the Moyal quantum algebra the involution is the
complex-conjugation as in the classical case. Thus in this case
quantum observables, exactly like classical observables, can be
chosen as real-valued functions. Further on we will deal only with
that distinguished quantization. Such a choice is not a
restriction as other quantization schemes known in the literature
are gauge equivalent to the Moyal one (see \cite{Blaszak:2012} and
\subsecref{subsec:2.5})).

Throughout the paper we will use the following notation. The \nbr{i}th
components of vectors $x,p \in \mathbb{R}^N$ representing positions and momenta
will be denoted by $x^i$ and $p_i$, i.e., $x = (x^1,\dotsc,x^N)$ and
$p = (p_1,\dotsc,p_N)$. By $xp$ we will understand the scalar product of vectors
$x$ and $p$, i.e., $xp = x^i p_i$. Often sequences of functions
$(Q^1,\dotsc,Q^N)$, $(P_1,\dotsc,P_N)$, derivatives $(\partial_{x^1},\dotsc,
\partial_{x^N})$, $(\partial_{p_1},\dotsc,\partial_{p_N})$ and operators
$(\hat{q}^1,\dotsc,\hat{q}^N)$, $(\hat{p}_1,\dotsc,\hat{p}_N)$ will be simply
denoted by $Q$, $P$, $\partial_x$, $\partial_p$, $\hat{q}$, $\hat{p}$. Also we
define products of these sequences with vectors and matrices in a standard way.

The Moyal quantization scheme is as follows. First, let us assume that
$M = \mathbb{R}^{2N}$ and $\mathcal{P} = \partial_{x^i} \wedge \partial_{p_i}$.
Define a \nbr{\star}product by a formula
\begin{equation*}
f \star g = f \exp \left( \frac{1}{2} i\hbar \overleftarrow{\partial}_{x^i}
    \overrightarrow{\partial}_{p_i} - \frac{1}{2} i\hbar
    \overleftarrow{\partial}_{p_i} \overrightarrow{\partial}_{x^i} \right) g.
\end{equation*}
This \nbr{\star}product is called the Moyal product. For a two-dimensional case
($N = 1$) the Moyal product reads
\begin{equation*}
f \star g = \sum_{k=0}^\infty \frac{1}{k!}
\left(\frac{i\hbar}{2}\right)^{k}
    \sum_{m=0}^k \binom{k}{m} (-1)^m
    (\partial_x^{k-m} \partial_p^m f)(\partial_x^m \partial_p^{k-m} g).
\end{equation*}
The deformed Poisson bracket $\lshad\sdot,\sdot\rshad$ associated
with the \nbr{\star}product will be given in terms of a
\nbr{\star}commutator $[\sdot,\sdot]$ as follows
\begin{equation*}
\lshad f,g \rshad = \frac{1}{i\hbar} [f,g]
= \frac{1}{i\hbar}(f \star g - g \star f), \quad
f,g \in \mathcal{A}_Q.
\end{equation*}

To avoid problems with convergence of the series in the above definition of the
\nbr{\star}product the common practise is to extend the space $C^\infty(M)$ to
a space $C^\infty(M)\lshad\hbar\rshad$ of formal power series in $\hbar$ with
coefficients from $C^\infty(M)$. The \nbr{\star}product is then properly defined
on such space. Other approach is to define a \nbr{\star}product on the space
$\mathcal{S}(\mathbb{R}^{2N})$ of Schwartz functions using an integral formula
and then extend it to an appropriate space of distributions \cite{Hansen:1984,
Kammerer:1986,Maillard:1986,Gracia-Bondia:1988}. It is then possible to define
a topology on such space of distributions for which the \nbr{\star}product, for
a certain class of functions, can be expressed as the above series.

Note, that the Moyal quantization was performed in a specific canonical
coordinate system, namely, in a Cartesian coordinate system on
$\mathbb{R}^{2N}$. In \secref{sec:3} we will see how to quantize systems in
different coordinate systems.

Let us prove a property of functions from $\mathcal{A}_Q$ which will be used
later on. For simplicity we only present the theorem for the case $N = 1$.
\begin{theorem}
\label{thm:2.1}
Every function $f \in \mathcal{A}_Q$ can be expanded into a \nbr{\star}power
series
\begin{equation}
f = \sum_{n,m = 0}^\infty a_{nm}
    \underbrace{x \star \dotsb \star x}_n \star
    \underbrace{p \star \dotsb \star p}_m,
\label{eq:2.1}
\end{equation}
where $a_{nm} \in \mathbb{C}$.
\end{theorem}
\begin{proof}
From recurrence relations
\begin{align*}
x \cdot (\underbrace{x \star \dotsb \star x}_n \star
    \underbrace{p \star \dotsb \star p}_m) & =
    \underbrace{x \star \dotsb \star x}_{n + 1} \star
    \underbrace{p \star \dotsb \star p}_m
- \frac{1}{2} i\hbar m \underbrace{x \star \dotsb \star x}_n
    \star \underbrace{p \star \dotsb \star p}_{m - 1}, \\
p \cdot (\underbrace{x \star \dotsb \star x}_n \star
    \underbrace{p \star \dotsb \star p}_m) & =
    \underbrace{x \star \dotsb \star x}_n \star
    \underbrace{p \star \dotsb \star p}_{m + 1}
- \frac{1}{2} i\hbar n \underbrace{x \star \dotsb \star x}_{n - 1}
    \star \underbrace{p \star \dotsb \star p}_m
\end{align*}
it follows that monomials $x^n p^m$ can be written as \nbr{\star}polynomials.
Thus after expanding $f$ into the Taylor series it is easily seen that $f$ can
be written in the form \eqref{eq:2.1}.
\end{proof}

\subsection{Space of states}
\label{subsec:2.2}
In general a space of states is fully characterized by the algebraic structure
of the quantum Poisson algebra $\mathcal{A}_Q$ \cite{Bordemann:1998,%
Waldmann:2005}. It can be shown that for the Moyal quantization states can
be represented as quantum distribution functions, i.e., square integrable
functions $\Psi$ defined on the phase space satisfying certain conditions
\cite{Dias:2004b,Blaszak:2012}. For this reason the Hilbert space $\mathcal{H} =
L^2(\mathbb{R}^{2N})$ of square integrable functions on the phase space will be
called a space of states. Observe, that the Moyal product can be extended to a
product between smooth functions from $C^\infty(\mathbb{R}^{2N})$ and square
integrable functions from $L^2(\mathbb{R}^{2N})$. It is also possible to define
the Moyal product between square integrable functions from
$L^2(\mathbb{R}^{2N})$ by extending it from the space
$\mathcal{S}(\mathbb{R}^{2N})$ of Schwartz functions \cite{Blaszak:2012}.

\subsection{Observables and states as operators}
\label{subsec:2.3}
In Moyal quantization scheme observables $A \in \mathcal{A}_Q$ and states
$\Psi \in \mathcal{H}$ can be treated as operators $\hat{A}$ and $\hat{\Psi}$
defined on the Hilbert space of states $\mathcal{H}$ by the following
prescription
\begin{equation*}
\hat{A} = A \star {}, \quad \hat{\Psi} = (2\pi\hbar)^{N/2} \Psi \star {}.
\end{equation*}
It can be proved that the operator $\hat{A}$ can be written as appropriately
ordered function $A$ of operators of position $(\hat{q}_M)^i = x^i \star {}
= x^i + \frac{1}{2} i\hbar \partial_{p_j}$ and momentum
$(\hat{p}_M)_j = p_j \star {} = p_j - \frac{1}{2} i\hbar \partial_{x^i}$
($[(\hat{q}_M)^i,(\hat{p}_M)_j] = i\hbar\delta^i_j$)
\begin{equation*}
\hat{A} = A \star {} = A_M(\hat{q}_M,\hat{p}_M).
\end{equation*}
$A_M(\hat{q},\hat{p})$ denotes symmetrically-ordered function of
operators $\hat{q}$ and $\hat{p}$.

Note that the Hermitian conjugation $\dagger$ of the operators $\hat{A}$ and
$\hat{\Psi}$ in Moyal quantization scheme is related to the complex-conjugation
$*$ (involution in $\mathcal{A}_Q$) by the formulae
\begin{equation*}
\hat{A}^\dagger = (A \star {})^\dagger = A^* \star {}, \quad
\hat{\Psi}^\dagger = ((2\pi\hbar)^{N/2} \Psi \star {})^\dagger
= (2\pi\hbar)^{N/2} \Psi^* \star {}.
\end{equation*}

\subsection{Expectation values of observables and time evolution equation}
\label{subsec:2.4}
Formulas for the expectation values of observables and the time evolution of
states are similar as in classical mechanics, except that the point-wise product
$\cdot$ of functions and the Poisson bracket $\{\sdot,\sdot\}$ have to be
replaced with the \nbr{\star}product and the quantum Poisson bracket
$\lshad\sdot,\sdot\rshad$. Thus, the expectation value of an observable
$A \in \mathcal{A}_Q$ in a state $\Psi \in \mathcal{H}$ is given by the formula
\begin{align*}
\braket{A}_{\Psi} & = \frac{1}{(2\pi\hbar)^{N/2}} \iint (A(0) \star
    \Psi(t))(x,p) \ud{x} \ud{p}
= \braket{A(0)}_{\Psi(t)} \\
& = \frac{1}{(2\pi\hbar)^{N/2}} \iint (A(t) \star \Psi(0))(x,p) \ud{x} \ud{p}
= \braket{A(t)}_{\Psi(0)}.
\end{align*}
The time evolution equation of quantum distribution functions $\Psi(t)$
(Schr{\"o}dinger picture) is the counterpart of the Liouville's equation
describing the time evolution of classical distribution functions, and is given
by the formula
\begin{equation*}
\frac{\dd \Psi}{\dd t} - \lshad H,\Psi(t) \rshad = 0
\quad \Leftrightarrow \quad
i\hbar \frac{\dd \Psi}{\dd t} - [H,\Psi(t)] = 0,
\end{equation*}
where $H$ is a Hamiltonian (distinguished real function from $\mathcal{A}_Q$).
The time evolution of quantum observable $A(t)$ (Heisenberg picture) is given by
\begin{equation}
\frac{\dd A}{\dd t} - \lshad A(t),H \rshad = 0
\quad \Leftrightarrow \quad
i\hbar \frac{\dd A}{\dd t} - [A(t),H] = 0.
\label{eq:2.2}
\end{equation}

\subsection{Equivalence of quantizations}
\label{subsec:2.5}
Two star-products $\star$ and $\star'$ are said to be gauge equivalent if there
exists a vector space automorphism $S \colon C^\infty(\mathbb{R}^{2N}) \to
C^\infty(\mathbb{R}^{2N})$ of the form
\begin{equation}
S = \sum_{k=0}^\infty \hbar^k S_k, \quad S_0 = 1,
\label{eq:2.3}
\end{equation}
where $S_k$ are linear operators, which satisfies the formula
\begin{equation*}
S(f \star g) = Sf \star' Sg \quad f,g \in C^\infty(\mathbb{R}^{2N}).
\end{equation*}
If, moreover, the automorphism $S$ preserves the deformed Poisson
brackets and involutions $*$ and $*'$ from the algebras
$\mathcal{A}_Q = (C^\infty(\mathbb{R}^{2N}), \star,
\lshad\sdot,\sdot\rshad, *)$ and $\mathcal{A}'_Q =
(C^\infty(\mathbb{R}^{2N}), \star', \lshad\sdot,\sdot\rshad',
*')$, i.e.,
\begin{align*}
S(\lshad f,g \rshad) & = \lshad Sf,Sg \rshad', \\
S(f^*) & = (Sf)^{*'},
\end{align*}
then $S$ is an isomorphism of the algebra $\mathcal{A}_Q$ onto the
algebra $\mathcal{A}'_Q$.

Two quantizations of a classical Hamiltonian system are equivalent
if there exists an isomorphism $S$ of their quantum Poisson
algebras. This equivalence is mathematical as well as physical. It
has been stressed out in \subsecref{subsec:2.1}, that to the same
measurable quantity correspond different functions from respective
quantum Poisson algebras. This observation seems to be missing in
considerations of different quantizations present in the
literature. In fact, to every observable $A \in \mathcal{A}_Q$
from one quantization scheme corresponds an observable $A' = SA
\in \mathcal{A}'_Q$ from the other quantization scheme. Both
observables $A$ and $A'$ describe the same measurable quantity and
in the limit $\hbar \to 0$ reduce to the same classical
observable. If, for example, the initial quantization is the Moyal
one, then quantum observables are real-valued functions and as a
quantum observable $A$ we can take a related classical observable
$A_C$, i.e., $A = A_C$ and is \nbr{\hbar}independent. Then, in
some other quantization scheme, gauge equivalent to the Moyal one,
to the same observable $A$ corresponds another function $A'=SA_C$,
$\hbar$-dependent and complex in general. Such approach to
equivalence of quantum systems introduces, indeed, physically
equivalent quantizations as the functions $A$, $A'$ from different
quantization schemes have the same spectra, expectation values,
etc., and when they are Hamiltonians they describe the same time
evolution.

The presented formalism of phase space quantum mechanics can be
extended to a very broad class of quantization schemes equivalent
to the Moyal-quantization scheme. Let $S$ be a vector space
automorphism of $C^\infty(\mathbb{R}^{2N})$ of the form
\eqref{eq:2.3}. Moreover, lets assume that $S$ satisfies
\begin{gather*}
Sx^i = x^i, \quad Sp_j = p_j, \\
S(f^*) = (Sf)^*.
\end{gather*}
The automorphism $S$ defines a new star-product equivalent to the Moyal-product
by the equation
\begin{equation*}
f \star_S g := S(S^{-1}f \star S^{-1}g), \quad f,g \in
C^\infty(\mathbb{R}^{2N}),
\end{equation*}
with involution being also the complex-conjugation. We will show
that such class of equivalent quantization schemes is related to
the choice of different quantum canonical coordinates for quantum
systems.

It is possible to define a morphism of spaces of states of different
quantization schemes, in terms of $S$. This morphism we will also denote by $S$.
In case when the initial quantization is the Moyal quantization $S$ will be a
Hilbert space isomorphism. In what follows we will restrict to the case when the
\nbr{S}image of the space of states $L^2(\mathbb{R}^{2N})$ is also a Hilbert
space $L^2(\mathbb{R}^{2N},\mu)$ of square integrable functions possibly with
respect to a different measure $\mu$, and when $S$ vanishes under the integral
sign, i.e.,
\begin{equation*}
\iint S\Psi(x,p) \ud{\mu(x,p)} = \iint \Psi(x,p) \ud{x}\ud{p},
\end{equation*}
for every quantum distribution function $\Psi$.

All previous results for the Moyal-quantization scheme are also valid for
the \nbr{S}quantization scheme. In particular, the symmetrically-ordered
operator function $A_M(\hat{q},\hat{p})$ of operators $\hat{q}^i$,
$\hat{p}_j$ has to be replaced by the following \nbr{S}ordering
\begin{equation*}
A_S(\hat{q},\hat{p}) := (S^{-1}A)_M(\hat{q},\hat{p}).
\end{equation*}
Also the following result holds true \cite{Blaszak:2012}
\begin{equation*}
A \star_S {} = A_S(\hat{q}_S,\hat{p}_S),
\end{equation*}
where
\begin{align*}
(\hat{q}_S)^i & := x^i \star_S {} = S (\hat{q}_M)^i S^{-1}, \\
(\hat{p}_S)_j & := p_j \star_S {} = S (\hat{p}_M)_j S^{-1}.
\end{align*}

The above results are important for transformation theory develop in the next
section.

\subsection{Ordinary description of quantum mechanics}
\label{subsec:2.6}
The phase space quantum mechanics is equivalent to the ordinary description
of quantum mechanics. To show this first note that the Hilbert space of
states can be written as the following tensor product of the Hilbert space
$L^2(\mathbb{R}^N)$ and a space dual to it $(L^2(\mathbb{R}^N))^*$:
\begin{equation*}
\mathcal{H} = \left(L^2(\mathbb{R}^N)\right)^* \otimes_S L^2(\mathbb{R}^N),
\end{equation*}
where the tensor product $\otimes_S$ is defined by
\begin{align*}
(\varphi^* \otimes_S \psi)(x,p) & = \frac{1}{(2\pi \hbar)^{N/2}} S
    \int \vd{y} e^{-\frac{i}{\hbar} py} \varphi^*\left(x - \frac{1}{2} y\right)
    \psi\left(x + \frac{1}{2} y\right) \\
& = \frac{1}{(2\pi \hbar)^{N/2}} \iiint \vd{x'} \vd{p'} \vd{y}
    \varphi^*\left(x' - \frac{1}{2} y\right) \psi\left(x' + \frac{1}{2} y\right)
    S(x,p,x',p') e^{-\frac{i}{\hbar} p'y},
\end{align*}
where $\varphi, \psi \in L^2(\mathbb{R}^N)$ and $S(x,p,x',p')$ is an integral
kernel of the isomorphism $S$.

States $\Psi \in \mathcal{H}$ treated as operators $\hat{\Psi} =
(2\pi\hbar)^{N/2} \Psi \star_S {}$ can be written in the following form
\begin{equation*}
\hat{\Psi} = \hat{1} \otimes_S \hat{\rho},
\end{equation*}
where $\hat{\rho}$ is some density operator. Hence to every pure or mixed
state $\Psi \in \mathcal{H}$ corresponds a unique density operator $\hat{\rho}$.
Similarly, observables $A \in \mathcal{A}_Q$ treated as operators $\hat{A} =
A \star_S {}$ take the form
\begin{equation*}
\hat{A} = A \star_S {}
= \hat{1} \otimes_S A_S(\hat{q},\hat{p}),
\end{equation*}
where $\hat{q}^i = x^i$ and $\hat{p}_j = -i\hbar\partial_{x^j}$ are standard
operators of position and momentum. In particular, from this it follows that
\begin{equation*}
A \star_S \Psi = \varphi^* \otimes_S A_S(\hat{q},\hat{p}) \psi,
\end{equation*}
for $\Psi = \varphi^* \otimes_S \psi$ and $\varphi,\psi \in L^2(\mathbb{R}^N)$.

The expectation values of observables $A \in \mathcal{A}_Q$ in states $\Psi
\in \mathcal{H}$ are the same as when computed in ordinary quantum mechanics
\begin{equation*}
\braket{A}_{\Psi} = \tr(\hat{\rho} A_S(\hat{q},\hat{p})),
\end{equation*}
where $\hat{\rho}$ is a density operator corresponding to $\Psi$ and
$A_S(\hat{q},\hat{p})$ is an operator corresponding to $A$. Also the
time evolution equation of states $\Psi \in \mathcal{H}$ corresponds to the
von Neumann equation describing the time evolution of density operators
$\hat{\rho}$.

\section{Canonical transformations of coordinates in phase space quantum
mechanics}
\label{sec:3}

\subsection{General theory}
\label{subsec:3.1}
As before we will consider the Moyal quantization of a classical Hamiltonian
system $(M,\mathcal{P},H)$, where $M = \mathbb{R}^{2N}$, $\mathcal{P} =
\partial_{x^i} \wedge \partial_{p_i}$, and $H \in C^\infty(M)$ is an arbitrary
real function.
A transformation of phase space coordinates of the quantum system is defined as
in classical mechanics, i.e., as a smooth bijective map $T \colon
\mathbb{R}^{2N} \supset U \ni (x',p') \to (x,p) \in W \subset \mathbb{R}^{2N}$.
For simplicity only transformations defined on almost the whole phase space will
be considered, i.e., it will be assumed that $\mathbb{R}^{2N} \setminus U$ and
$\mathbb{R}^{2N} \setminus W$ are sets of the Lebesgue-measure zero. A
transformation $T$ transforms every observable $A \in C^\infty(\mathbb{R}^{2N})$
and every state $\Psi \in L^2(\mathbb{R}^{2N})$ by the formulae
\begin{equation*}
A' = A \circ T, \quad \Psi' = \Psi \circ T.
\end{equation*}
In addition, also the Moyal-product needs to be transformed. The transformed
star-product, denoted hereafter by $\star'_T$, should fulfill the following
natural condition
\begin{equation*}
(f \star g) \circ T = (f \circ T) \star'_T (g \circ T), \quad
f,g \in \mathcal{A}_Q.
\end{equation*}
The \nbr{\star'_T}product is, in fact, given by the following formula
\begin{equation*}
f \star'_T g = f \exp \left( \frac{1}{2} i\hbar \overleftarrow{D_{x'^i}}
    \overrightarrow{D_{p'_j}} - \frac{1}{2} i\hbar \overleftarrow{D_{p'_j}}
    \overrightarrow{D_{x'^i}} \right) g,
\end{equation*}
where vector fields $D_{x'^i}$, $D_{p'_j}$ are derivations $\partial_{x^i}$,
$\partial_{p_j}$ transformed by the transformation $T$ according to the rule
\begin{align*}
(\partial_{x^i} f) \circ T & = D_{x'^i}(f \circ T), \quad
    f \in C^\infty(\mathbb{R}^{2N}),\\
(\partial_{p_j} f) \circ T & = D_{p'_j}(f \circ T), \quad
    f \in C^\infty(\mathbb{R}^{2N}).
\end{align*}
Notice, that a new product defined by vector fields $D_{x'^i}$, $D_{p'_j}$ is
well defined star-product, which associativity follows from commutativity
of $D_{x'^i}$ and $D_{p'_j}$.

To this moment we considered general transformations of
coordinates. In what follows we will focus on an important class
of transformations, namely, canonical transformations. In
classical mechanics a canonical transformation is such
transformation $T$ of phase space coordinates which transforms the
system from one canonical coordinate system to the other. In other
words, $T$ is a canonical transformation if it preserves the form
of the Poisson bracket, i.e.,
\begin{equation*}
\{x'^i, p'_j\}' = \delta^i_j,
\end{equation*}
where $\{\sdot,\sdot\}'$ denotes a Poisson bracket transformed by $T$ to the new
coordinate system:
\begin{equation*}
\{f,g\}' = \{f \circ T^{-1}, g \circ T^{-1}\} \circ T, \quad
f,g \in C^\infty(\mathbb{R}^{2N}).
\end{equation*}

Canonical transformations in phase space quantum mechanics are
defined in a similar manner. Namely, a quantum canonical
transformation is such transformation $T$ of coordinates which
preserves the form of the deformed Poisson bracket, i.e.,
\begin{equation*}
\lshad x'^i, p'_j \rshad' = \delta^i_j,
\end{equation*}
where $\lshad\sdot,\sdot\rshad'$ denotes a deformed Poisson bracket transformed
by $T$ to the new coordinate system:
\begin{equation*}
\lshad f,g \rshad' = \lshad f \circ T^{-1}, g \circ T^{-1} \rshad \circ T, \quad
f,g \in C^\infty(\mathbb{R}^{2N}).
\end{equation*}

In what follows we will consider only transformations $T$ which are quantum
canonical transformations, with a special attention to the subclass of these
which are also classically canonical, because the theory takes the simplest form
for such class of transformations.

Now we take the advantage from the following property of a transformed
star-product, valid for a wide class of quantum canonical transformations as we
will see in \secref{sec:5}. For any quantum canonical coordinate transformation
$T$ with respect to a given \nbr{\star}product, there exists a unique linear
automorphism $S_T$ of the algebra $\mathcal{A}_Q$ in the form \eqref{eq:2.3}
satisfying
\begin{subequations}
\label{eq:3.1}
\begin{gather}
S_T(f \star' g) = (S_T f) \star'_T (S_T g), \quad f,g \in \mathcal{A}_Q,
\label{eq:3.1a} \\
S_T x'^i = x'^i, \quad S_T p'_j = p'_j,
\label{eq:3.1b} \\
S_T(f^*) = (S_T f)^*, \quad f \in \mathcal{A}_Q,
\label{eq:3.1c}
\end{gather}
\end{subequations}
where $\star'$ is the initial product expressed in $x',p'$ variables
\begin{equation*}
f \star' g = f \exp \left( \frac{1}{2} i\hbar \overleftarrow{\partial}_{x'^i}
    \overrightarrow{\partial}_{p'_i} - \frac{1}{2} i\hbar
    \overleftarrow{\partial}_{p'_i} \overrightarrow{\partial}_{x'^i} \right) g.
\end{equation*}
We will consider only such transformations for which the above statement is
true, however, we believe that this statement holds for every quantum canonical
transformation. The above property will play a crucial role while passing to the
ordinary description of quantum mechanics. From \eqref{eq:3.1} it is clear that
two quantum Hamiltonian systems related to each other by some quantum canonical
change of phase space coordinates are gauge equivalent.

For deriving the form of the automorphism $S_T$ for particular transformations
$T$, other form of the conditions \eqref{eq:3.1} will be more useful.

\begin{theorem}
\label{thm:3.1}
The conditions \eqref{eq:3.1} are fulfilled if and only if the conditions
\begin{subequations}
\label{eq:3.3}
\begin{align}
(\hat{q}'_T)^i & = S_T (\hat{q}_M)^i S_T^{-1}, \label{eq:3.3a} \\
(\hat{p}'_T)_j & = S_T (\hat{p}_M)_j S_T^{-1}, \label{eq:3.3b}
\end{align}
\end{subequations}
are fulfilled, where $(\hat{q}'_T)^i = x'^i \star'_T {}$,
$(\hat{p}'_T)_j = p'_j \star'_T {}$,
$(\hat{q}_M)^i = x'^i \star' {}$ and $(\hat{p}_M)_j =
p'_j \star' {}$.
\end{theorem}

\begin{proof}
For simplicity we will present the proof for a two-dimensional case ($N = 1$).
If the conditions \eqref{eq:3.1} are fulfilled then trivially the conditions
\eqref{eq:3.3} are fulfilled. Assume now, that the conditions \eqref{eq:3.3} are
fulfilled. From \eqref{eq:3.3} it follows that \eqref{eq:3.1a} will be satisfied
for every $f$ in the form of a \nbr{\star'}monomial $x' \star' \dotsb \star' x'
\star' p' \star' \dotsb \star' p'$ as
\begin{align*}
& S_T((\underbrace{x' \star' \dotsb \star' x'}_n \star'
\underbrace{p' \star' \dotsb \star' p'}_m) \star' g)
= S_T x' \star'_T S_T(\underbrace{x' \star' \dotsb \star' x'}_{n - 1} \star'
\underbrace{p' \star' \dotsb \star' p'}_m \star' g) \\
& \qquad = \underbrace{S_T x' \star'_T \dotsb \star'_T S_T x'}_n \star'_T
\underbrace{S_T p' \star'_T \dotsb \star'_T S_T p'}_m \star'_T S_T g \\
& \qquad = \underbrace{S_T x' \star'_T \dotsb \star'_T S_T x'}_n \star'_T
\underbrace{S_T p' \star'_T \dotsb \star'_T S_T p'}_{m - 2} \star'_T
S_T(p' \star' p') \star'_T S_T g \\
& \qquad = S_T(\underbrace{x' \star' \dotsb \star' x'}_n \star'
\underbrace{p' \star' \dotsb \star' p'}_m) \star'_T S_T g.
\end{align*}
From the linearity of $S_T$ and the fact that the general $f \in \mathcal{A}_Q$
can be written in the form of the series \eqref{eq:2.1} follows the condition
\eqref{eq:3.1a}. The condition \eqref{eq:3.1b} can be received by calculating
left and right sides of \eqref{eq:3.3} on function identically equal 1.
Furthermore, the condition \eqref{eq:3.1c} can be proved in a similar manner as
\eqref{eq:3.1a} noting that the complex-conjugation is an involution for the
$\star'$ and \nbr{\star'_T}products, and using \eqref{eq:3.1b}.
\end{proof}

Note that every transformation $T$ of coordinates preserves the involution
of the algebra $\mathcal{A}_Q$ (complex conjugation). Indeed, there holds
\begin{equation*}
(f \circ T)^* = f^* \circ T, \quad f \in \mathcal{A}_Q.
\end{equation*}
We prove the property \eqref{eq:3.1c} of the isomorphism $S_T$. Namely,

\begin{theorem}
For every quantum canonical transformation $T$ the isomorphism $S_T \colon
(C^\infty(\mathbb{R}^{2N}),\star',*) \to (C^\infty(\mathbb{R}^{2N}),\star'_T,*)$
preserves the involution, i.e.,
\begin{equation*}
S_T(f^*) = (S_T f)^*, \quad f \in C^\infty(\mathbb{R}^{2N}).
\end{equation*}
In other words $S_T = S^*_T$.
\end{theorem}

\begin{proof}
Using the fact that $S_T$ is an isomorphism of the space of states $\mathcal{H}$
onto the transformed space of states $\mathcal{H}'$ (see below) the following
equality holds
\begin{align*}
\braket{\Phi|f^* \star' \Psi}_{\mathcal{H}} & =
    \braket{\Phi|(f \star' {})^\dagger \Psi}_{\mathcal{H}}
= \braket{f \star' \Phi|\Psi}_{\mathcal{H}}
= \braket{S_T(f \star' \Phi)|S_T \Psi}_{\mathcal{H}'}
= \braket{S_T f \star'_T S_T \Phi|S_T \Psi}_{\mathcal{H}'} \\
& = \braket{S_T \Phi|(S_T f \star'_T {})^\dagger S_T \Psi}_{\mathcal{H}'}
= \braket{S_T \Phi|(S_T f)^* \star'_T S_T \Psi}_{\mathcal{H}'}
= \braket{\Phi|S_T^{-1}(S_T f)^* \star' \Psi}_{\mathcal{H}}
\end{align*}
for every $\Phi,\Psi \in \mathcal{H}$. From this follows that
$f^* = S_T^{-1}(S_T f)^*$.
\end{proof}

It means that for every quantum canonical transformation $T$ the
isomorphism $S_T$ generates new quantization scheme, equivalent to
the Moyal one, with involution being also complex conjugation.

A transformation $T$ of coordinates transforms the Hilbert space of states
$L^2(\mathbb{R}^{2N})$ into the Hilbert space $L^2(\mathbb{R}^{2N},\mu_T)$ with
the standard scalar product, where $\dd{\mu_T(x',p')} = \abs{\det T'(x',p')}
\ud{x'}\ud{p'}$. Of course, if $T$ is also a classical canonical transformation
then $\dd{\mu_T(x',p')} = \dd{x'}\ud{p'}$, since the Jacobian of a classical
canonical transformation is equal 1. The following property of $S_T$ has to hold
to ensure the consistency of the whole approach to transformations of
coordinates:

\begin{quote}
For a transformation $T$ of coordinates
\begin{equation*}
S_T \colon L^2(\mathbb{R}^{2N}) \to L^2(\mathbb{R}^{2N},\mu_T)
\end{equation*}
and $S_T$ is a Hilbert space isomorphism. Hence the space
$L^2(\mathbb{R}^{2N},\mu_T)$ is a properly defined space of states after the
transformation of coordinates. Moreover, the isomorphism $S_T$ vanishes under
the integral sign, i.e.,
\begin{equation*}
\iint f(x',p') \ud{x'}\ud{p'} = \iint S_T f(x',p') \ud{\mu(x',p')}.
\end{equation*}
\end{quote}

Thus we will consider only such transformations for which the above property is
fulfilled, however, we believe that every quantum canonical transformation
satisfies the above property. Observe that a classically canonical
transformation $T$ transforms the space of states $L^2(\mathbb{R}^{2N})$ onto
itself. Thus $S_T$ in this case is an isomorphism (unitary operator) on
$L^2(\mathbb{R}^{2N})$.

Canonical transformations appear as trajectories in phase space when one
considers the time evolution of a system. This is true both in classical and
quantum mechanics. In phase space quantum mechanics the solution of quantum
Hamiltonian equations
\begin{equation}
\dot{Q}^i(t) = \lshad Q^i(t),H \rshad, \quad
\dot{P}_j(t) = \lshad P_j(t),H \rshad,
\label{eq:3.4}
\end{equation}
where $Q^i(x,p,0) = x^i$ and $P_j(x,p,0) = p_j$, i.e., the Heisenberg
representation \eqref{eq:2.2} for observables of position and momentum,
generates a quantum flow $\Phi_t$ in phase space according to
an equation
\begin{equation}
\Phi_t(x,p;\hbar) = (Q(x,p,t;\hbar),P(x,p,t;\hbar)).
\label{eq:3.5}
\end{equation}
For every instance of time $t$ the map $\Phi_t$ is a quantum canonical
transformation (quantum symplectomorphism) from coordinates $x,p$ to new
coordinates $x' = Q(x,p,t;\hbar), p' = P(x,p,t;\hbar)$. Of course, in the limit
$\hbar \to 0$ any quantum trajectory reduces to a classical
trajectory. More thorough description of quantum trajectories we present in
\cite{Blaszak:2012b}.

\subsection{Transformations which are both classical and quantum canonical}
\label{subsec:3.3}
The simplest subclass of quantum canonical transformations is generated by these
quantum flows which coincide with classical flows. It means that the
quantum canonical transformation \eqref{eq:3.5} is $\hbar$ independent and
coincides with classical symplectomorphism, i.e., is a simultaneous solution of
quantum and classical Hamiltonian equations. Such transformations are generated
from an appropriate generating function for any fixed value of the evolution
parameter $t$.

In what follows we will present four important classes of nonlinear canonical
transformations (both classical and quantum). The linear case is considered
separately in \secref{sec:5}. A sufficient condition for such class of
transformations is a linearity in one set of arguments of an appropriate
generating function. Let us begin with the case of two-dimensional phase space.
The first class of transformations is determined by generating function of the
form
\begin{equation*}
F_1(x,x') = x\phi_1(x') + \phi_2(x')
\end{equation*}
with related transformations expressed by the equations
\begin{equation*}
p = \frac{\partial F_1}{\partial x}(x,x'), \quad
p' = -\frac{\partial F_1}{\partial x'}(x,x'),
\end{equation*}
where $\phi_1$ is some smooth bijective function and $\phi_2$ is some smooth
function. The above equations give the transformation in a form
\begin{equation}
T_1(x',p') = \Bigl(-\bigl(\phi'_1(x')\bigr)^{-1} p'
    - \bigl(\phi'_1(x')\bigr)^{-1} \phi'_2(x'), \phi_1(x')\Bigr),
\label{eq:3.11}
\end{equation}
where $\phi'(x')=\frac{d}{dx'}\phi(x')$. Note, that a generating
function
\begin{equation*}
\tilde{F}_1(x,x') = -x'\phi_1(x) - \phi_2(x)
\end{equation*}
generates a transformation $\tilde{T}_1$ being an inverse transformation to
$T_1$.

The second class is generated by a function
\begin{equation*}
F_2(p,p') = -p\phi_1(p') - \phi_2(p')
\end{equation*}
with a related transformation expressed by the equations
\begin{equation*}
x = -\frac{\partial F_2}{\partial p}(p,p'), \quad
x' = \frac{\partial F_2}{\partial p'}(p,p'),
\end{equation*}
which give
\begin{equation*}
T_2(x',p') = \Bigl(\phi_1(p'), -\bigl(\phi'_1(p')\bigr)^{-1} x'
    - \bigl(\phi'_1(p')\bigr)^{-1} \phi'_2(p')\Bigr).
\end{equation*}
In this case the generating function
\begin{equation*}
\tilde{F}_2(p,p') = p'\phi_1(p) + \phi_2(p)
\end{equation*}
generates a transformation $\tilde{T}_2$ being an inverse transformation to
$T_2$.

The third and fourth classes of transformations are generated by functions
\begin{align*}
F_3(x,p') & = x\phi_1(p') + \phi_2(p'), \\
F_4(x',p) & = -p\phi_1(x') - \phi_2(x'),
\end{align*}
with related transformations expressed by the equations
\begin{equation*}
p = \frac{\partial F_3}{\partial x}(x,p'), \quad
x'= \frac{\partial F_3}{\partial p'}(x,p'),
\end{equation*}
and
\begin{equation*}
x = -\frac{\partial F_4}{\partial p}(x',p), \quad
p'= -\frac{\partial F_4}{\partial x'}(x',p).
\end{equation*}
The above equations give the transformations in a form
\begin{subequations}
\label{eq:3.14}
\begin{align}
T_3(x',p') & = \Bigl(\bigl(\phi'_1(p')\bigr)^{-1} x'
    - \bigl(\phi'_1(p')\bigr)^{-1} \phi'_2(p'), \phi_1(p')\Bigr),
\label{eq:3.14a} \\
T_4(x',p') & = \Bigl(\phi_1(x'), \bigl(\phi'_1(x')\bigr)^{-1} p'
    - \bigl(\phi'_1(x')\bigr)^{-1} \phi'_2(x')\Bigr).
\label{eq:3.14b}
\end{align}
\end{subequations}
In this case the generating functions
\begin{align*}
\tilde{F}_3(x,p') & = -p'\phi_1(x) - \phi_2(x), \\
\tilde{F}_4(x',p) & = x'\phi_1(p) + \phi_2(p)
\end{align*}
generate transformations $\tilde{T}_3 = T_4^{-1}$ and $\tilde{T}_4 = T_3^{-1}$.

Note, that the transformations $T_2$, $T_3$ and $T_4$ can be constructed
from $T_1$ with the help of an interchange of variables transformation
$I(x',p') = (-p', x')$ being a special case of the transformation $T_1$
generated by a function $F(x,x') = xx'$:
\begin{subequations}
\label{eq:3.12}
\begin{align}
T_2 & = I \circ T_1 \circ I^{-1},  \label{eq:3.12a} \\
T_3 & = T_1 \circ I^{-1},          \label{eq:3.12b} \\
T_4 & = I^{-1} \circ T_1.          \label{eq:3.12c}
\end{align}
\end{subequations}
Thus instead of considering the transformation theory for transformations
$T_1$, $T_2$, $T_3$ and $T_4$ it is enough to consider only the transformation
$T_1$ and it appropriate compositions with $I$.

The four presented classes of transformations are obviously classically
canonical. It can be shown that they are also quantum canonical.

The presented considerations can be easily extended to a \nbr{2N}dimensional
case. As an example let us present the transformation \eqref{eq:3.11} in
\nbr{2N}dimensions. Consider a generating function
\begin{equation}
F(x,x') = \sum_i x^i \phi_1^i(x') + \phi_2(x'),
\label{eq:3.13}
\end{equation}
where $\phi_1 = (\phi^1,\dotsc,\phi^N) \colon \mathbb{R}^N \to \mathbb{R}^N$ is
a smooth bijective function and $\phi_2 \colon \mathbb{R}^N \to \mathbb{R}$ a
smooth function, and $x'= (x'^1,\dotsc,x'^N)$. The function $F$ generates the
following transformation
\begin{align*}
T(x',p') & = \left(-\bigl(\phi'_1(x')\bigr)^{-1} \cdot p' -
    \bigl(\phi'_1(x')\bigr)^{-1} \cdot \phi'_2(x'), \phi_1(x')\right) \\
& = \Bigl(-\bigl(J_{\phi_1}^{-1}(x')\bigr)^i_1 p'_i
    - \bigl(J_{\phi_1}^{-1}(x')\bigr)^i_1 \bigl(J_{\phi_2}(x')\bigr)_i,\dotsb,
    - \bigl(J_{\phi_1}^{-1}(x')\bigr)^i_N p'_i
    - \bigl(J_{\phi_1}^{-1}(x')\bigr)^i_N \bigl(J_{\phi_2}(x')\bigr)_i, \\
& \qquad \phi_1^1(x'), \dotsc, \phi_1^N(x') \Bigr),
\end{align*}
where $(J_{\phi_1})_i^j = \frac{\partial \phi_1^j}{\partial x'^i}$ is a Jacobian
matrix of the function $\phi_1 = (\phi_1^1,\dotsc,\phi_1^N)$ and
$(J_{\phi_2})_i = \frac{\partial \phi_2}{\partial x'^i}$ is a Jacobian matrix of
$\phi_2$. A calculation shows that this transformation is also quantum
canonical. All other cases can be extended in a similar fashion.

\subsection{Transformation of observables treated as operators}
\label{subsec:3.4}
Let us introduce the following operator $\hat{T} \colon \mathcal{H} \to
\mathcal{H}'$ of the Hilbert spaces of states $\mathcal{H} =
L^2(\mathbb{R}^{2N})$ and $\mathcal{H}' = L^2(\mathbb{R}^{2N},\mu_T)$ by
the equation
\begin{equation*}
\hat{T} \Psi = \Psi \circ T, \quad \Psi \in \mathcal{H}.
\end{equation*}
The operator $\hat{T}$ is obviously linear and bijective with the inverse
given by
\begin{equation*}
\hat{T}^{-1} \Psi = \Psi \circ T^{-1}, \quad \Psi \in \mathcal{H}'.
\end{equation*}
Moreover, the operator $\hat{T}$ preserves the scalar product, hence it is
an isomorphism of the Hilbert spaces of states. There holds the following
theorem.

\begin{theorem}
Let $T$ be some quantum canonical transformation of coordinates and $A \in
\mathcal{A}_Q$. There holds
\begin{equation*}
\hat{T} A_M(\hat{q}_M, \hat{p}_M) \hat{T}^{-1} =
    A'_{S_T}(\hat{q}'_T, \hat{p}'_T)
\end{equation*}
where $A'= A \circ T$ and
\begin{equation*}
(\hat{q}'_T)^i = x'^i \star'_T {}, \quad (\hat{p}'_T)_j = p'_j \star'_T {}.
\end{equation*}
\end{theorem}

\begin{proof}
Let $\Psi \in \mathcal{H}$. Then there holds
\begin{equation*}
\hat{T} A_M(\hat{q}_M, \hat{p}_M) \Psi = \hat{T}(A \star \Psi)
= A' \star'_T \hat{T} \Psi = A'_{S_T}(\hat{q}'_T, \hat{p}'_T) \hat{T} \Psi.
\end{equation*}
\end{proof}

\begin{corollary}
\begin{align*}
(\hat{q}'_T)^i & = \hat{T} (Q^i)_M(\hat{q}_M, \hat{p}_M) \hat{T}^{-1}, \\
(\hat{p}'_T)_j & = \hat{T} (P_j)_M(\hat{q}_M, \hat{p}_M) \hat{T}^{-1},
\end{align*}
or from the other side
\begin{align*}
(\hat{q}_M)^i & = \hat{T}^{-1} (q^i)_{S_T}(\hat{q}'_T, \hat{p}'_T) \hat{T}, \\
(\hat{p}_M)_j & = \hat{T}^{-1} (p_j)_{S_T}(\hat{q}'_T, \hat{p}'_T) \hat{T},
\end{align*}
where $T^{-1}(x,p) = (Q(x,p), P(x,p))$ and $T(x',p') = (q(x',p'), p(x',p'))$.
\end{corollary}

\section{Transformations of coordinates in ordinary quantum mechanics}
\label{sec:4}

\subsection{General theory}
\label{subsec:4.1}
In the following section we use the formalism of transformations of coordinates
in phase space quantum mechanics, presented in previous section, to develop a
theory of transformations in ordinary quantum mechanics. First, observe that the
transformation $T$ of coordinates induces a unitary operator $\hat{U}_T
\colon L^2(\mathbb{R}^N) \to L^2(\mathbb{R}^N)$ defined on the Hilbert space of
states of the ordinary quantum mechanics, which transforms vector states and
observables to the new coordinate system. Such operator is defined by the
following equation
\begin{equation}
\hat{T}(\varphi^* \otimes_M \psi) = (\varphi^* \otimes_M \psi) \circ T
=: (\hat{U}_T \varphi)^* \otimes_{S_T} \hat{U}_T \psi, \quad
\varphi,\psi \in L^2(\mathbb{R}^N).
\label{eq:4.5}
\end{equation}

The operator $\hat{U}_T$ should be unitary to assure that after a
transformation of coordinates states remain properly defined. We state this
as an another assumption about the considered class of transformations.

From the definition of the operator $\hat{U}_T$ follows an important formula
for the transformation of observables to the new coordinate system.

\begin{theorem}
For $A \in \mathcal{A}_Q$ and an arbitrary quantum canonical transformation $T$
of coordinates there holds
\begin{equation*}
A'_{S_T}(\hat{q}', \hat{p}') = \hat{U}_T A_M(\hat{q}, \hat{p}) \hat{U}_T^{-1},
\end{equation*}
where $A'= A \circ T$ and
\begin{align*}
(\hat{q}^i \psi)(x) & = x^i \psi(x), & (\hat{p}_j \psi)(x)
& = -i\hbar \frac{\partial \psi}{\partial x^j}(x), \\
(\hat{q}'^i \psi')(x') & = x'^i \psi'(x'), & (\hat{p}'_j \psi')(x')
& = -i\hbar \frac{\partial \psi'}{\partial x'^j}(x'),
\end{align*}
for $\psi,\psi' \in L^2(\mathbb{R}^N)$.
\end{theorem}

\begin{proof}
Let $\Psi = \varphi^* \otimes_M \psi$. Then from one side
\begin{equation*}
(A \star \Psi) \circ T = A' \star'_T \hat{T} \Psi
= (\hat{U}_T \varphi)^* \otimes_{S_T}
    A'_{S_T}(\hat{q}', \hat{p}') \hat{U}_T \psi
\end{equation*}
and from the other side
\begin{align*}
(A \star \Psi) \circ T & = (\varphi^* \otimes_M
    A_M(\hat{q}, \hat{p}) \psi) \circ T
= (\hat{U}_T \varphi)^* \otimes_{S_T}
    \hat{U}_T A_M(\hat{q}, \hat{p}) \psi \\
& = (\hat{U}_T \varphi)^* \otimes_{S_T}
    \hat{U}_T A_M(\hat{q}, \hat{p}) \hat{U}_T^{-1} \hat{U}_T \psi.
\end{align*}
Comparison of the above two formulae implies the result.
\end{proof}

The above result shows that applying the Born's quantization rule to a
transformed classical observable gives an operator unitarily equivalent with an
operator corresponding to an untransformed classical observable, provided that
the ordering of $\hat{q}'$, $\hat{p}'$ will be appropriately changed. This
ensures the consistency of a standard quantization procedure since we can
quantize a classical Hamiltonian system in whatever canonical coordinate system,
provided that we use an appropriate ordering of operators of position and
momentum.

\begin{corollary}
\begin{subequations}
\label{eq:4.1}
\begin{align}
\hat{q}'^i & = \hat{U}_T (Q^i)_M(\hat{q},\hat{p}) \hat{U}_T^{-1},
\label{eq:4.1a} \\
\hat{p}'_j & = \hat{U}_T (P_j)_M(\hat{q},\hat{p}) \hat{U}_T^{-1},
\label{eq:4.1b}
\end{align}
\end{subequations}
or from the other side
\begin{align*}
\hat{q}^i & = \hat{U}_T^{-1} (q^i)_{S_T}(\hat{q}',\hat{p}') \hat{U}_T, \\
\hat{p}_j & = \hat{U}_T^{-1} (p_j)_{S_T}(\hat{q}',\hat{p}') \hat{U}_T,
\end{align*}
where $T^{-1}(x,p) = (Q(x,p), P(x,p))$ and $T(x',p') = (q(x',p'), p(x',p'))$.
\end{corollary}

The above result can be understood as follows. An operator of position
corresponding to a new coordinate system can be simply defined as
$Q_M(\hat{q},\hat{p})$. This operator can be written in a position
representation, i.e., as an operator of multiplication by coordinate variable.
More precisely, there exist a unitary operator $\hat{U}$ such that
\begin{equation*}
\hat{U} Q_M(\hat{q}, \hat{p}) \hat{U}^{-1} = x'.
\end{equation*}
From previous result $\hat{U} = \hat{U}_T$.

\subsection{The form of the operator $\hat{U}_T$}
\label{subsec:4.2}
\Eqsref{eq:4.1} can be used to derive the form of the operator
$\hat{U}_T$ for considered class of transformations. Indeed, for a given
transformation $T$ it is enough to find such unitary operator $\hat{U}_T$
which will satisfy \eqref{eq:4.1}. Such operator will be unique.

In what follows let us find the form of the unitary operator $\hat{U}_T$, in a
two-dimensional case, for a transformation $T_1(x',p') = (-(\phi'_1(x'))^{-1} p'
- (\phi'_1(x'))^{-1} \phi'_2(x'), \phi_1(x'))$ generated by a function
$F_1(x,x') = x\phi_1(x') + \phi_2(x')$ from \subsecref{subsec:3.3}. One
calculates that for $\hat{U}_{T_1}$ in the form
\begin{equation}
(\hat{U}_{T_1} \varphi)(x') = \frac{1}{\sqrt{2\pi\hbar}} \int \varphi(x)
    \abs{\phi'_1(x')}^{1/2} e^{-\frac{i}{\hbar}(x\phi_1(x') + \phi_2(x'))}\ud{x}
\label{eq:4.2}
\end{equation}
the inverse operator to $\hat{U}_{T_1}$ reads
\begin{equation}
(\hat{U}_{T_1}^{-1} \varphi)(x) = \frac{1}{\sqrt{2\pi\hbar}} \int
\varphi(x') \abs{\phi'_1(x')}^{1/2} e^{\frac{i}{\hbar}(x\phi_1(x')
+ \phi_2(x'))}\ud{x'}
 \label{eq:4.2a}
\end{equation}
and that \eqref{eq:4.1} are indeed satisfied. To prove this note that
$Q_M(\hat{q},\hat{p}) = \phi_1^{-1}(\hat{p})$. Then
\begin{align*}
(\hat{U}_{T_1} Q_M(\hat{q},\hat{p}) \hat{U}_{T_1}^{-1} \varphi)(x') & =
    \frac{1}{2\pi\hbar} \iint \phi_1^{-1}(-i\hbar\partial_x) \left( \varphi(x'')
    \abs{\phi'_1(x'')}^{1/2} e^{\frac{i}{\hbar}(x\phi_1(x'') + \phi_2(x''))}
    \right) \\
& \quad \times \abs{\phi'_1(x')}^{1/2}
    e^{-\frac{i}{\hbar}(x\phi_1(x') + \phi_2(x'))} \ud{x}\ud{x''} \\
& = \frac{1}{2\pi\hbar} \iint x'' \varphi(x'')
    \abs{\phi'_1(x'') \phi'_1(x')}^{1/2}
    e^{\frac{i}{\hbar}(\phi_1(x'') - \phi_1(x'))x}
    e^{\frac{i}{\hbar}(\phi_2(x'') - \phi_2(x'))} \ud{x}\ud{x''} \\
& = \int x'' \varphi(x'') \abs{\phi'_1(x'') \phi'_1(x')}^{1/2}
    \delta(\phi_1(x'') - \phi_1(x'))
    e^{\frac{i}{\hbar}(\phi_2(x'') - \phi_2(x'))} \ud{x''}.
\end{align*}
After introducing new variable $\tilde{x}'' = \phi_1(x'')$ one gets that
$\dd{x''} = \abs{\phi'_1(\phi_1^{-1}(\tilde{x}''))}^{-1} \ud{\tilde{x}''}$
and the above formula can be written in a form
\begin{align*}
(\hat{U}_{T_1} Q_M(\hat{q},\hat{p}) \hat{U}_{T_1}^{-1} \varphi)(x') & =
    \int \phi_1^{-1}(\tilde{x}'') \varphi(\phi_1^{-1}(\tilde{x}''))
    \abs{\phi'_1(\phi_1^{-1}(\tilde{x}''))}^{1/2} \abs{\phi'_1(x')}^{1/2}
    \delta(\tilde{x}'' - \phi_1(x')) \\
& \quad\times e^{\frac{i}{\hbar}(\phi_2(\phi_1^{-1}(\tilde{x}'')) - \phi_2(x'))}
    \ud{\tilde{x}''} \\
& = x'\varphi(x').
\end{align*}
\Eqref{eq:4.1b} can be proved analogically. Note, that
\eqref{eq:4.2} and \eqref{eq:4.2a} can be written in a general
form
\begin{equation}
(\hat{U}_{T_1} \varphi)(x') = \frac{1}{\sqrt{2\pi\hbar}} \int \varphi(x)
    \sqrt{\Abs{\frac{\partial^2 F_1}{\partial x \partial x'}(x,x')}}
    e^{-\frac{i}{\hbar} F_1(x,x')} \ud{x},
\label{eq:4.3}
\end{equation}
\begin{equation}
(\hat{U}_{T_1}^{-1} \varphi)(x) = \frac{1}{\sqrt{2\pi\hbar}} \int
\varphi(x')
    \sqrt{\Abs{\frac{\partial^2 F_1}{\partial x \partial x'}(x,x')}}
    e^{\frac{i}{\hbar} F_1(x,x')} \ud{x'}.
\label{eq:4.3a}
\end{equation}
A direct calculations show that $\hat{U}_{T_1}$ is indeed a
unitary operator on the Hilbert space $L^2(\mathbb{R})$.

Let us also find the form of the unitary operator $\hat{U}_T$, also in a
two-dimensional case, for the point transformation $T_4(x',p') = (\phi_1(x'),
(\phi'_1(x')^{-1} p' - (\phi'_1(x')^{ -1} \phi'_2(x'))$ generated by the
function $F_4(x',p) = -p\phi_1(x') - \phi_2(x')$. One calculates that
$\hat{U}_{T_4}$ in the form
\begin{align}
(\hat{U}_{T_4} \varphi)(x') & = \frac{1}{\sqrt{2\pi\hbar}} \int
    \tilde{\varphi}(p)
    \sqrt{\Abs{\frac{\partial^2 F_4}{\partial x' \partial p}(x',p)}}
    e^{\frac{i}{\hbar} F_4(x',p)} \ud{p}
= \frac{1}{\sqrt{2\pi\hbar}} \int \tilde{\varphi}(p) \abs{\phi'_1(x')}^{1/2}
    e^{-\frac{i}{\hbar}(p\phi_1(x') + \phi_2(x'))} \ud{p} \nonumber \\
& = \abs{\phi'_1(x')}^{1/2} e^{-\frac{i}{\hbar} \phi_2(x')} \varphi(\phi_1(x')),
\label{eq:4.4}
\end{align}
where $\tilde{\varphi}$ denotes the inverse Fourier transform of $\varphi$,
indeed satisfies \eqref{eq:4.1}. The inverse operator to $\hat{U}_{T_4}$ is
then given by
\begin{align*}
(\hat{U}_{T_4}^{-1} \varphi)(x) & =
\frac{1}{\sqrt{2\pi\hbar}}\mathcal{F}^{-1} \int
    \varphi(x')
    \sqrt{\Abs{\frac{\partial^2 F_4}{\partial x' \partial p}(x',p)}}
    e^{-\frac{i}{\hbar} F_4(x',p)} \ud{x'} \\
    & =\frac{1}{2\pi\hbar} \iint \varphi(x')
    \abs{\phi'_1(x')}^{1/2} e^{\frac{i}{\hbar}(p\phi_1(x') + \phi_2(x') - xp)}
    \ud{x'}\ud{p} \\
& = \abs{\phi'_1(\phi_1^{-1}(x))}^{-1/2}
    e^{\frac{i}{\hbar}\phi_2(\phi_1^{-1}(x))} \varphi(\phi_1^{-1}(x)).
\end{align*}
One can show that, in accordance with \eqref{eq:3.12c}, the operator
$\hat{U}_{T_4}$ can be constructed from the operator $\hat{U}_{T_1}$ from
\eqref{eq:4.2} and the Fourier transform
\begin{equation*}
\hat{U}_{T_4} = \hat{U}_{I^{-1} \circ T_1} = \hat{U}_{T_1} \hat{U}_I^{-1}
= \hat{U}_{T_1} \mathcal{F}^{-1},
\end{equation*}
because $\hat{U}_I$ is exactly the Fourier transform $\mathcal{F}$. In
analogical way are constructed unitary operators $\hat{U}_{T_2}$ and
$\hat{U}_{T_3}$:
\begin{align*}
\hat{U}_{T_2} & = \hat{U}_{I \circ T_1 \circ I^{-1}}
= \hat{U}_I^{-1} \hat{U}_{T_1} \hat{U}_I
= \mathcal{F}^{-1} \hat{U}_{T_1} \mathcal{F}, \\
\hat{U}_{T_3} & = \hat{U}_{T_1 \circ I^{-1}} = \hat{U}_I^{-1} \hat{U}_{T_1}
= \mathcal{F}^{-1} \hat{U}_{T_1}.
\end{align*}

For a \nbr{2N}dimensional case and the transformation $T$ given by the
generating function \eqref{eq:3.13}, the respective operator $\hat{U}_T$ takes
the form
\begin{align*}
(\hat{U}_T \varphi)(x') & = \frac{1}{(2\pi\hbar)^{N/2}} \int \varphi(x)
    \abs{\det \phi'_1(x')}^{1/2}
    e^{-\frac{i}{\hbar}(\sum_i x^i \phi_1^i(x') + \phi_2(x'))} \ud{x} \\
& = \frac{1}{(2\pi\hbar)^{N/2}} \int \varphi(x)
    \sqrt{\Abs{\det\left[\frac{\partial^2 F}{\partial x^i \partial x'^j}(x,x')
    \right]}} e^{-\frac{i}{\hbar} F(x,x')} \ud{x}.
\end{align*}

\section{Examples}
\label{sec:5}

\subsection{Linear transformations}
\label{subsec:5.1}
In what follows the well known linear transformations of coordinates of a
quantum phase space will be reconsidered in the frame of formalism developed in
previous sections. The linear transformation is a transformation $T \colon
\mathbb{R}^2 \to \mathbb{R}^2$ given by the equation
\begin{equation*}
T(x',p') = (dx' - bp', -cx' + ap'),
\end{equation*}
where $a,b,c,d \in \mathbb{R}$. Moreover, it is assumed that $ad - bc = 1$,
which makes this transformation canonical both in a classical and quantum
phase space, i.e., it preserves both Poisson bracket and star-commutator. The
inverse transformation is given by the following equation
\begin{equation*}
T^{-1}(x,p) = (ax + bp, cx + dp).
\end{equation*}

The linear transformation $T$ is generated by a function $F(x,x') =
\frac{1}{b} xx' - \frac{a}{2b} x^2 - \frac{d}{2b} x'^2$, i.e.,
\begin{equation*}
p = \frac{\partial F}{\partial x}(x,x'), \quad
p' = -\frac{\partial F}{\partial x'}(x,x'),
\end{equation*}
where $(x,p) = T(x',p')$.

For a given function $f \in C^\infty(\mathbb{R}^2)$ the derivatives of the
function $f$ transform as follows
\begin{align*}
\frac{\partial f}{\partial x} \circ T & = a \frac{\partial}{\partial x'}
    (f \circ T) + c \frac{\partial}{\partial p'} (f \circ T), \\
\frac{\partial f}{\partial p} \circ T & = b \frac{\partial}{\partial x'}
    (f \circ T) + d \frac{\partial}{\partial p'} (f \circ T).
\end{align*}
Using the above formulae one easily finds that the linear transformation $T$
preserves the \nbr{\star}product, i.e., the \nbr{\star}product do not change
after the transformation of coordinates
\begin{equation*}
(f \star g) \circ T = (f \circ T) \star' (g \circ T), \quad
f,g \in C^{\infty}(\mathbb{R}^2).
\end{equation*}
Of course, in this case the isomorphism $S_T = 1$.

Now, let us prove by direct calculation that the unitary operator
$\hat{U}_T \colon L^2(\mathbb{R}) \to L^2(\mathbb{R})$ of the Hilbert space
$L^2(\mathbb{R})$ induced by $T$ indeed is of the form \eqref{eq:4.3}. In other
words, let us show that $\hat{U}_T$ given by \eqref{eq:4.3} satisfies the
formula
\begin{equation}
\hat{T}(\varphi^* \otimes_M \psi) = (\varphi^* \otimes_M \psi) \circ T
= (\hat{U}_T \varphi)^* \otimes_{M} \hat{U}_T \psi,
\label{eq:5.1}
\end{equation}
for every $\varphi,\psi \in L^2(\mathbb{R})$.

Indeed, the twisted tensor product of $\varphi^*$
and $\psi$ takes the form
\begin{equation*}
(\varphi^* \otimes_{M} \psi)(x,p) = \frac{1}{\sqrt{2\pi\hbar}} \int \vd{y}
    \varphi^*\left(x - \frac{1}{2}y\right) \psi\left(x + \frac{1}{2} y\right)
    e^{-\frac{i}{\hbar} py}.
\end{equation*}
Using the above formula the left side of \eqref{eq:5.1} takes the form
\begin{align*}
\hat{T}(\varphi^* \otimes_{M} \psi)(x',p') & =
    (\varphi^* \otimes_{M} \psi)(T(x',p')) \\
& = \frac{1}{\sqrt{2\pi\hbar}} \int \vd{y}
    \varphi^*\left(dx' - bp' - \frac{1}{2} y\right)
    \psi\left(dx' - bp' + \frac{1}{2} y\right)
    e^{-\frac{i}{\hbar} (-cx' + ap')y}.
\end{align*}
The right side of \eqref{eq:5.1} takes the form
\begin{align*}
((\hat{U}_T \varphi)^* \otimes_{M} \hat{U}_T \psi)(x',p') & =
    \frac{1}{\sqrt{2\pi\hbar}} \int \vd{y'}
    (\hat{U}_T \varphi)^*\left(x' - \frac{1}{2} y'\right)
    (\hat{U}_T \psi)\left(x' + \frac{1}{2} y'\right) e^{-\frac{i}{\hbar}p'y'} \\
& = \frac{1}{(2\pi\hbar)^{3/2} \abs{b}} \iiint \vd{x} \vd{y} \vd{y'}
    \varphi^*(x) \psi(y) e^{\frac{i}{\hbar} F(x, x' - \frac{1}{2} y')}
    e^{-\frac{i}{\hbar} F(y, x' + \frac{1}{2} y')} e^{-\frac{i}{\hbar} p'y'} \\
& = \frac{1}{(2\pi\hbar)^{3/2} \abs{b}} \iint \vd{x}\vd{y} \varphi^*(x) \psi(y)
    e^{-\frac{i}{\hbar}(\frac{1}{b}(y - x)x' - \frac{a}{2b}(y^2 - x^2))} \\
& \quad \times \int \vd{y'} e^{-\frac{i}{\hbar}(\frac{1}{2b}(x + y)
    - \frac{d}{b}x' + p')y'}.
\end{align*}
After changing the variables under the integral sign from $(x,y)$ to
$(x - \frac{1}{2} y, x + \frac{1}{2} y)$ and from $y'$ to $by'$ the above
equation takes the form
\begin{align*}
((\hat{U}_T \varphi)^* \otimes_{M} \hat{U}_T \psi)(x',p') & =
    \frac{1}{(2\pi\hbar)^{3/2}} \iint \vd{x}\vd{y}
    \varphi^*\left(x - \frac{1}{2}y\right) \psi\left(x + \frac{1}{2} y\right)
    e^{-\frac{i}{\hbar}(\frac{1}{b}yx' - \frac{a}{b}xy)} \\
& \quad \times \int \vd{y'} e^{-\frac{i}{\hbar}(x - dx' + bp')y'} \\
& = \frac{1}{\sqrt{2\pi\hbar}} \iint \vd{x}\vd{y}
    \varphi^*\left(x - \frac{1}{2}y\right) \psi\left(x + \frac{1}{2} y\right)
    e^{-\frac{i}{\hbar}(\frac{1}{b}yx' - \frac{a}{b}xy)}
    \delta(x - dx' + bp') \\
& = \frac{1}{\sqrt{2\pi\hbar}} \int \vd{y}
    \varphi^*\left(dx' - bp' - \frac{1}{2} y\right)
    \psi\left(dx' - bp' + \frac{1}{2} y\right)
    e^{-\frac{i}{\hbar} (-cx' + ap')y}.
\end{align*}

Linear transformations of coordinates appear when considering quantum
trajectories generated by Hamiltonian functions, quadratic with respect to
canonical coordinates. Take as an example the Hamiltonian of harmonic oscillator
\begin{equation*}
H(x,p) = \frac{1}{2} \left( p^2 + \omega^2 x^2 \right).
\end{equation*}
It happens that in such case the quantum trajectory coincides with
the classical one. Thus a classical/quantum trajectory $\Phi_t =
(Q(t),P(t))$ of a harmonic oscillator takes the form
\begin{equation*}
Q(x,p,t) = x \cos \omega t + \omega^{-1} p \sin \omega t, \quad
P(x,p,t) = p \cos \omega t - \omega x \sin \omega t.
\end{equation*}
One can immediately see that for every $t$ the above equations
define a linear canonical transformation with generating function
\begin{equation*}
F(x,x')=\frac{\omega}{\sin \omega t}xx'-\frac{1}{2}\omega \cot
\omega t \, x^2-\frac{1}{2}\omega \cot \omega t \, x'^2.
\end{equation*}

Concluding, the above results show that applying the Born's
quantization rule to any classical observable transformed by
canonical linear coordinate transformation, gives an operator
unitary equivalent to an operator corresponding to untransformed
classical observable, with the same ordering of $\hat{q}'$ and
$\hat{p}'$.

\subsection{A class of nonlinear transformations}
\label{subsec:5.2}
In the following example a class of nonlinear transformations of coordinates of
a quantum phase space will be considered. The transformation in question will be
taken in the form
\begin{equation*}
T(x',p') = (-ap' - a\phi'(x'), a^{-1}x').
\end{equation*}
The transformation $T$ is generated by the function $F(x,x') = a^{-1} xx' +
\phi(x')$ ($a \in \mathbb{R}$, $a \neq 0$, $\phi$ being an arbitrary smooth
function). This transformation is a classical canonical transformation. For a
given function $f \in C^\infty(\mathbb{R}^2)$ the derivatives of the function
$f$ transform as follows
\begin{align*}
\frac{\partial f}{\partial x} \circ T & =
    -\frac{1}{a} \frac{\partial}{\partial p'}(f \circ T), \\
\frac{\partial f}{\partial p} \circ T & =
    a\frac{\partial}{\partial x'} (f \circ T)
    - a\phi''(x') \frac{\partial}{\partial p'}(f \circ T).
\end{align*}
Hence the Moyal \nbr{\star}product transforms to the following
product
\begin{equation*}
f \star'_T g = f \exp \left( \frac{1}{2} i\hbar \overleftarrow{D_{x'}}
    \overrightarrow{D_{p'}} - \frac{1}{2} i\hbar \overleftarrow{D_{p'}}
    \overrightarrow{D_{x'}} \right) g,
\end{equation*}
where
\begin{align*}
D_{x'} & = a^{-1} \partial_{p'}, \\
D_{p'} & = a \partial_{x'} - a\phi''(x') \partial_{p'}.
\end{align*}

One can calculate that
\begin{align*}
\hat{q}'_T & = x' \star'_T {} = x' + \frac{1}{2} i\hbar \partial_{p'}, \\
\hat{p}'_T & = p' \star'_T {}
= p' - \frac{1}{2} i\hbar \partial_{x'} - \sum_{n = 2}^\infty \frac{1}{n!}
    \left( \frac{i\hbar}{2} \right)^n \phi^{(n + 1)}(x') \partial_{p'}^{n},
\end{align*}
and that
\begin{equation*}
[\hat{q}'_T, \hat{p}'_T] = i\hbar.
\end{equation*}
The above equation shows that the transformation $T$ is a quantum canonical
transformation.

The isomorphism $S_T$ from \eqref{eq:3.1} intertwining the Moyal
\nbr{\star'}product with the \nbr{\star'_T}product is given by
\begin{align*}
S_T & = \exp \left( -\sum_{n = 1}^\infty \frac{1}{(2n + 1)!} (-1)^n
    \left( \frac{\hbar}{2} \right)^{2n} \phi^{(2n + 1)}(x')
    \partial_{p'}^{2n + 1} \right) \\
& = \exp \left( \frac{i}{\hbar}
    \left( \phi \left(x' + \frac{1}{2}i\hbar \partial_{p'} \right)
    - \phi \left(x'- \frac{1}{2}i\hbar \partial_{p'} \right)
    - i\hbar \phi'(x') \partial_{p'} \right) \right) \\
& = \exp \left( \frac{i}{\hbar} \left( \phi(\hat{q}_M) - \phi'(x') \hat{q}_M
    - \phi(\hat{q}^*_M) + \phi'(x') \hat{q}^*_M \right) \right).
\end{align*}

Indeed, from \thmref{thm:3.1} it is enough to prove that
\begin{align*}
\hat{q}'_T & = S_T \hat{q}_M S_T^{-1}, \\
\hat{p}'_T & = S_T \hat{p}_M S_T^{-1}.
\end{align*}
Since $S_T = e^{\hat{A}}$ where $\hat{A} = -\sum_{n = 1}^\infty
\frac{1}{(2n + 1)!} (-1)^n \left(\frac{\hbar}{2}\right)^{2n} \phi^{(2n + 1)}(x')
\partial_{p'}^{2n + 1}$ the above equations, from the Hadamard's lemma, take the
form
\begin{subequations}
\label{eq:5.2}
\begin{align}
\hat{q}'_T & = e^{[\hat{A}, \sdot]} \hat{q}_M,
\label{eq:5.2a} \\
\hat{p}'_T & = e^{[\hat{A}, \sdot]} \hat{p}_M.
\label{eq:5.2b}
\end{align}
\end{subequations}
One easily calculates that $[\hat{A}, \hat{q}_M] = 0$. Thus
\begin{equation*}
e^{[\hat{A}, \sdot]} \hat{q}_M = \hat{q}_M = \hat{q}'_T,
\end{equation*}
which proves \eqref{eq:5.2a}.

On the other hand, one finds that
\begin{align*}
[\hat{A}, \hat{p}_M] & = -\sum_{n = 1}^\infty \frac{1}{(2n)!} (-1)^n
    \left( \frac{\hbar}{2} \right)^{2n} \phi^{(2n + 1)}(x') \partial_{p'}^{2n}
    - \frac{1}{2} i\hbar \sum_{n = 1}^\infty \frac{1}{(2n + 1)!} (-1)^n
    \left(\frac{\hbar}{2}\right)^{2n} \phi^{(2n + 2)}(x') \partial_{p'}^{2n + 1}
\\
& = -\sum_{n = 2}^\infty \frac{1}{n!} \left( \frac{i\hbar}{2} \right)^n
    \phi^{(n + 1)}(x') \partial_{p'}^{n}
\end{align*}
and that
\begin{equation*}
[\hat{A}, [\hat{A}, \hat{p}_M]] = 0.
\end{equation*}
Thus
\begin{equation*}
e^{[\hat{A}, \sdot]} \hat{p}_M = p' - \frac{1}{2} i\hbar
    \partial_{x'} - \sum_{n = 2}^\infty \frac{1}{n!}
    \left( \frac{i\hbar}{2} \right)^n \phi^{(n + 1)}(x') \partial_{p'}^{n}
= \hat{p}'_T,
\end{equation*}
which proves \eqref{eq:5.2b}.

To derive the second formula on $S_T$ from the first one note that
\begin{align*}
S_T & = \exp \left( -\sum_{n = 1}^\infty \frac{1}{(n + 1)!}
    \left( \frac{i\hbar}{2} \right)^n \frac{1 + (-1)^n}{2} \phi^{(n + 1)}(x')
    \partial_{p'}^{n + 1} \right) \\
& = \exp \left( \frac{i}{\hbar} \sum_{n = 2}^\infty \frac{1}{n!}
    \left( 1 - (-1)^n \right) \phi^{(n)}(x')
    \left( \frac{1}{2} i\hbar \partial_{p'} \right)^n \right) \\
& = \exp \left( \frac{i}{\hbar} \left( \phi \left(x' + \frac{1}{2} i\hbar
    \partial_{p'}\right) - \phi(x') - \phi'(x') \frac{1}{2} i\hbar \partial_{p'}
    - \phi \left(x' - \frac{1}{2} i\hbar \partial_{p'} \right) + \phi(x')
    - \phi'(x') \frac{1}{2} i\hbar \partial_{p'} \right) \right) \\
& = \exp \left( \frac{i}{\hbar} \left( \phi \left(x' + \frac{1}{2} i\hbar
    \partial_{p'} \right) - \phi \left(x' - \frac{1}{2} i\hbar \partial_{p'}
    \right)
    - i\hbar \phi'(x') \partial_{p'} \right) \right).
\end{align*}

Observe, that the map $S_T$ preserves the involution, i.e.,
\begin{equation*}
S_T(f^*) = (S_T f)^*, \quad f \in \mathcal{A}_Q,
\end{equation*}
which can be immediately seen from the fact that $S_T = S^*_T$.

In accordance with the assumption made in \subsecref{subsec:3.1} concerning
$S_T$ the map $S_T$ is an isomorphism (unitary operator) of the space of states
$\mathcal{H} = L^2(\mathbb{R}^2)$ onto the transformed space of states
$\mathcal{H}' = L^2(\mathbb{R}^2)$.

Indeed, an operator of the form $\exp \left( \frac{i}{\hbar}
f(x', i\hbar \partial_{p'}) \right)$ ($f \colon \mathbb{R}^2 \to \mathbb{R}$
being some smooth function) transforms square integrable functions into square
integrable functions. In fact, this operator can be written in the following
form
\begin{equation*}
\exp \left( \frac{i}{\hbar} f(x', i\hbar \partial_{p'}) \right) =
    \mathcal{F}_p^{-1} \exp \left( \frac{i}{\hbar} f(x', \eta) \right)
    \mathcal{F}_p,
\end{equation*}
as a composition of operators transforming square integrable functions into
square integrable functions. Hence $S_T \colon L^2(\mathbb{R}^2) \to
L^2(\mathbb{R}^2)$. The map $S_T$ is obviously linear and injective, moreover,
it is easy to show that $\braket{\Phi | S_T \Psi}_{L^2} =
\braket{S_T^{-1} \Phi | \Psi}_{L^2}$, thus $S_T$ is an isomorphism of
$L^2(\mathbb{R}^2)$ onto $L^2(\mathbb{R}^2)$.

Now, let us prove by direct calculation that the unitary operator
$\hat{U}_T \colon L^2(\mathbb{R}) \to L^2(\mathbb{R})$ of the Hilbert space
$L^2(\mathbb{R})$ induced by $T$ indeed is of the form \eqref{eq:4.3}. In other
words, let us show that $\hat{U}_T$ given by \eqref{eq:4.3} satisfies the
formula
\begin{equation}
\hat{T}(\varphi^* \otimes_M \psi) = (\varphi^* \otimes_M \psi) \circ T
= (\hat{U}_T \varphi)^* \otimes_{S_T} \hat{U}_T \psi,
\label{eq:5.3}
\end{equation}
for every $\varphi,\psi \in L^2(\mathbb{R})$.

Indeed, the left side of \eqref{eq:5.3} takes the form
\begin{equation*}
(\varphi^* \otimes_M \psi)(T(x',p')) = \frac{1}{\sqrt{2\pi\hbar}}
    \int \vd{y'} \varphi^*\left(-ap' - a\phi'(x') - \frac{1}{2} y'\right)
    \psi\left(-ap' - a\phi'(x') + \frac{1}{2} y'\right)
    e^{-\frac{i}{\hbar} a^{-1}x'y'}.
\end{equation*}
The right side of \eqref{eq:5.3} takes the form
\begin{align}
((\hat{U}_T \varphi)^* \otimes_{S_T} \hat{U}_T \psi)(x',p') & =
    S_T \frac{1}{\sqrt{2\pi\hbar}} \int \vd{y''}
    (\hat{U}_T\varphi)^*\left(x' - \frac{1}{2} y''\right)
    (\hat{U}_T\psi)\left(x' + \frac{1}{2} y''\right) e^{-\frac{i}{\hbar}p'y''}
    \nonumber \\
& = S_T \frac{1}{(2\pi\hbar)^{3/2}} |a|^{-1} \iiint \vd{y''} \vd{x} \vd{y}
    \varphi^*(x) e^{\frac{i}{\hbar}(a^{-1}x(x' - \frac{1}{2} y'')
    + \phi(x' - \frac{1}{2} y''))} \nonumber \\
& \quad \times \psi(y) e^{-\frac{i}{\hbar}(a^{-1}y(x' + \frac{1}{2} y'')
    + \phi(x' + \frac{1}{2} y''))} e^{-\frac{i}{\hbar} p'y''}.
\label{eq:5.4}
\end{align}
Since
\begin{align*}
& \exp \left( \frac{i}{\hbar} \left( \phi \left(x' + \frac{1}{2} i\hbar
    \partial_{p'} \right) - \phi \left(x' - \frac{1}{2} i\hbar \partial_{p'}
    \right)
    - i\hbar \phi'(x') \partial_{p'} \right) \right)
    g(x') e^{-\frac{i}{\hbar} p'y''} = \\
& \qquad = \exp \left( \frac{i}{\hbar} \left( \phi \left(x' + \frac{1}{2} y''
    \right) - \phi \left(x' - \frac{1}{2} y'' \right)
    - \phi'(x') y'' \right) \right) g(x') e^{-\frac{i}{\hbar} p'y''},
\end{align*}
\eqref{eq:5.4} takes the form
\begin{align*}
& ((\hat{U}_T \varphi)^* \otimes_{S_T} \hat{U}_T \psi)(x',p') =
    \frac{1}{(2\pi\hbar)^{3/2}} |a|^{-1} \iiint \vd{y''} \vd{x} \vd{y}
    \varphi^*(x) \psi(y) \\
& \qquad \times \exp \left( \frac{i}{\hbar} \left( a^{-1}(x - y)x'
    - a^{-1}\left(\frac{1}{2} x + \frac{1}{2}y\right)y''
    + \phi\left(x' - \frac{1}{2} y''\right)
    - \phi\left(x' + \frac{1}{2} y''\right) \right. \right. \\
& \qquad {} - \phi\left(x' - \frac{1}{2} y''\right)
    + \phi\left(x' + \frac{1}{2}y''\right)
    - \phi'(x')y'' - p'y'' \bigr) \biggr) \\
& \quad = \frac{1}{(2\pi\hbar)^{3/2}} |a|^{-1} \iiint \vd{y''} \vd{x} \vd{y}
    \varphi^*(x) \psi(y) \\
& \qquad \times \exp \left( \frac{i}{\hbar} \left( -a^{-1}(y - x)x'
    - \left(a^{-1}\left(\frac{1}{2} x + \frac{1}{2} y\right) + \phi'(x')
    + p'\right)y''\right)\right).
\end{align*}
Introducing new variables $y' = y - x$ and $z' = a^{-1}(\frac{1}{2} x +
\frac{1}{2} y) + \phi'(x') + p'$ the above equation can be written in the form
\begin{align*}
& ((\hat{U}_T \varphi)^* \otimes_{S_T} \hat{U}_T \psi)(x',p') =
    \frac{1}{(2\pi\hbar)^{3/2}} \iiint \vd{y''} \vd{y'} \vd{z'}
    \varphi^*\left(az' - a\phi'(x') - ap' - \frac{1}{2}y'\right) \\
& \qquad \times \psi\left(az' - a\phi'(x') - ap' + \frac{1}{2} y'\right)
    e^{-\frac{i}{\hbar}ax'y'} e^{-\frac{i}{\hbar}z'y''} \\
& \quad = \frac{1}{\sqrt{2\pi\hbar}} \int \vd{y'}
    \varphi^*\left(-ap' - a\phi'(x') - \frac{1}{2} y'\right)
    \psi\left(-ap' - a\phi'(x') + \frac{1}{2} y'\right)
    e^{-\frac{i}{\hbar} a^{-1}x'y'}.
\end{align*}

\subsection{Nonlinear point transformations}
\label{subsec:5.3}
Let us consider in the following example the point transformations of phase
space coordinates
\begin{equation*}
T(x',p') = (\phi(x'),(\phi'(x'))^{-1}p')
\end{equation*}
generated by functions $F(x',p) = -p\phi(x')$ ($\phi$ being an
arbitrary smooth bijective function). These transformations from
construction are classical canonical transformations. For a given
function $f \in C^\infty(\mathbb{R}^2)$ the derivatives of the
function $f$ transform as follows
\begin{align*}
\frac{\partial f}{\partial x} \circ T & = (\phi'(x'))^{-2} \phi''(x') p'
    \frac{\partial}{\partial p'}(f \circ T)
    + (\phi'(x'))^{-1} \frac{\partial}{\partial x'}(f \circ T), \\
\frac{\partial f}{\partial p} \circ T & = \phi'(x')
    \frac{\partial}{\partial p'}(f \circ T).
\end{align*}
Hence the Moyal \nbr{\star}product transforms to the following product
\begin{equation*}
f \star'_T g = f \exp \left( \frac{1}{2} i\hbar \overleftarrow{D_{x'}}
    \overrightarrow{D_{p'}} - \frac{1}{2} i\hbar \overleftarrow{D_{p'}}
    \overrightarrow{D_{x'}} \right) g,
\end{equation*}
where
\begin{align*}
D_{x'} & = (\phi'(x'))^{-1} \partial_{x'}
    + (\phi'(x'))^{-2} \phi''(x') p' \partial_{p'}, \\
D_{p'} & = \phi'(x') \partial_{p'}.
\end{align*}

One can calculate that
\begin{align*}
\hat{q}'_T & = x' \star'_T {} = x' + \frac{1}{2}i\hbar\partial_{p'} \\
& \quad {} + \sum_{n=1}^\infty \frac{1}{(n+1)!}
    \left( \frac{i\hbar}{2} \right)^{n+1}
    (\phi'(x'))^{-n} \sum_{(k_1,\dotsc,k_n) \in G_n} a_{k_1,\dotsc,k_n}^{(n)}
    \phi^{(k_1)}(x') \dotsm \phi^{(k_n)}(x') \partial_{x'}\partial_{p'}^{n+1},\\
\hat{p}'_T & = p' \star'_T {} = p' - \frac{1}{2}i\hbar\partial_{x'} \\
& \quad {} + \sum_{n=1}^\infty \frac{1}{(n+1)!}
    \left( \frac{i\hbar}{2} \right)^{n+1}
    (\phi'(x'))^{-n-1} \left( \sum_{(k_1,\dotsc,k_{n+1}) \in G_{n+1}}
    b_{k_1,\dotsc,k_{n+1}}^{(n+1)} \phi^{(k_1)}(x') \dotsm \phi^{(k_{n+1})}(x')
    p'\partial_{p'}^{n+1} \right. \\
& \quad {} + \sum_{(k_1,\dotsc,k_n) \in G_n} c_{k_1,\dotsc,k_n}^{(n)}
    \phi^{(k_1)}(x') \dotsm \phi^{(k_n)}(x') \phi'(x') \partial_{x'}
    \partial_{p'}^n \\
& \quad {} \left. + \sum_{(k_1,\dotsc,k_n) \in G_n} n c_{k_1,\dotsc,k_n}^{(n)}
    \phi^{(k_1)}(x') \dotsm \phi^{(k_n)}(x') \phi''(x') \partial_{p'}^n \right),
\end{align*}
where $G_n = \{(k_1,\dotsc,k_n) \in \mathbb{N}^n \colon 1 \le k_1 \le \dotsb
\le k_n, \sum_{i=1}^n k_i = 2n\}$ and $a_{k_1,\dotsc,k_n}^{(n)}$,
$b_{k_1,\dotsc,k_n}^{(n)}$, $c_{k_1,\dotsc,k_n}^{(n)}$ are some integer
constants. To the third order in $\hbar$ the operators $\hat{q}'_T$ and
$\hat{p}'_T$ take the form
\begin{align*}
\hat{q}'_T & = x' + \frac{1}{2}i\hbar\partial_{p'}
    + \frac{1}{2!} \left( \frac{i\hbar}{2} \right)^2 (\phi'(x'))^{-1}
    \Bigl( -\phi''(x') \Bigr) \partial_{p'}^2 \\
& \quad {} + \frac{1}{3!} \left( \frac{i\hbar}{2} \right)^3 (\phi'(x'))^{-2}
    \Bigl( 3(\phi''(x'))^2 - \phi'(x') \phi'''(x') \Bigr) \partial_{p'}^3
    + o(\hbar^4), \\
\hat{p}'_T & = p' - \frac{1}{2}i\hbar\partial_{x'}
    + \frac{1}{2!} \left( \frac{i\hbar}{2} \right)^2 (\phi'(x'))^{-2} \biggl(
    \Bigl( -3(\phi''(x'))^2 + \phi'(x') \phi'''(x') \Bigr) p'\partial_{p'}^2 \\
& \quad {} + \Bigl( -2\phi''(x') \Bigr) \phi'(x') \partial_{x'} \partial_{p'}
+ \Bigl( -2\phi''(x') \Bigr) \phi''(x') \partial_{p'} \biggr) \\
& \quad {} + \frac{1}{3!} \left( \frac{i\hbar}{2} \right)^3 (\phi'(x'))^{-3}
    \biggl( \Bigl( 6(\phi''(x'))^3 - 7\phi'(x')\phi''(x')\phi'''(x')
    + (\phi'(x'))^2 \phi''''(x') \Bigr) p'\partial_{p'}^3 \\
& \quad {} + \Bigl( 3(\phi''(x'))^2 - 3\phi'(x')\phi'''(x') \Bigr) \phi'(x')
    \partial_{x'} \partial_{p'}^2
    + 2\Bigl( 3(\phi''(x'))^2 - 3\phi'(x')\phi'''(x') \Bigr) \phi''(x')
    \partial_{p'}^2 \biggr) + o(\hbar^4).
\end{align*}
Moreover, it can be verified that
\begin{equation*}
[\hat{q}'_T, \hat{p}'_T] = i\hbar,
\end{equation*}
which shows that $T$ is a quantum canonical transformation.

Using \thmref{thm:3.1} it can be shown that the isomorphism $S_T$ is of the
form
\begin{align*}
S_T & = \exp \left[ \sum_{n=1}^\infty \frac{1}{(2n+1)!}
\left( \frac{\hbar}{2} \right)^{2n} (\phi'(x'))^{-2n} \left(
\sum_{(k_1,\dotsc,k_{2n}) \in G_{2n}} A_{k_1,\dotsc,k_{2n}}^{(2n)}
    \phi^{(k_1)}(x') \dotsm \phi^{(k_{2n})}(x') p'\partial_{p'}^{2n+1}
    \right. \right. \\
& \quad {} + \sum_{(k_1,\dotsc,k_{2n-1}) \in G_{2n-1}}
    B_{k_1,\dotsc,k_{2n-1}}^{(2n-1)} \phi^{(k_1)}(x') \dotsm
    \phi^{(k_{2n-1})}(x') \phi'(x') \partial_{x'} \partial_{p'}^{2n} \\
& \quad {} \left. \left. + \sum_{(k_1,\dotsc,k_{2n-1}) \in G_{2n-1}}
    C_{k_1,\dotsc,k_{2n-1}}^{(2n-1)} \phi^{(k_1)}(x') \dotsm
    \phi^{(k_{2n-1})}(x') \phi''(x') \partial_{p'}^{2n} \right) \right],
\end{align*}
where $A_{k_1,\dotsc,k_n}^{(n)}$, $B_{k_1,\dotsc,k_n}^{(n)}$,
$C_{k_1,\dotsc,k_n}^{(n)}$ are some integer constants. To the second order in
$\hbar$ the isomorphism $S_T$ reads
\begin{align*}
S_T & = 1 + \frac{1}{3!} \left( \frac{\hbar}{2} \right)^2 (\phi'(x'))^{-2}
    \biggl(\Bigl(3(\phi''(x'))^2 - \phi'(x')\phi'''(x') \Bigr) p'\partial_{p'}^3
    + \Bigl( 3\phi''(x') \Bigr) \phi'(x') \partial_{x'} \partial_{p'}^2 \\
& \quad {} + \Bigl( 3\phi''(x') \Bigr) \phi''(x') \partial_{p'}^2 \biggr)
    + o(\hbar^4).
\end{align*}

There is no simple expression for the constants $A_{k_1,\dotsc,k_n}^{(n)}$,
$B_{k_1,\dotsc,k_n}^{(n)}$, $C_{k_1,\dotsc,k_n}^{(n)}$ although, they can be
calculated for particular cases of point transformations. One of such cases is
for $\phi \colon \mathbb{R} \setminus \{0\} \to \mathbb{R} \setminus \{0\}$,
$\phi(x') = \sgn(x') \sqrt{\abs{2x'}}$ (the example from Introduction).
The operators $\hat{q}'_T$ and $\hat{p}'_T$ take then the form
\begin{align*}
\hat{q}'_T & = x' + \frac{1}{2}i\hbar\partial_{p'}
    - \frac{1}{8} \hbar^2 \sgn(x') \abs{2x'}^{-1} \partial_{p'}^2, \\
\hat{p}'_T & = p' - \frac{1}{2}i\hbar\partial_{x'}
    + \sum_{n=1}^\infty \left( -\frac{i\hbar}{2} \right)^{n+1}
    \Bigl( (\sgn(x'))^n \abs{2x'}^{-n} \partial_{x'} \partial_{p'}^n
    - n(\sgn(x'))^{n+1} \abs{2x'}^{-n-1} \partial_{p'}^n \Bigr),
\end{align*}
and the isomorphism $S_T$ is expressed by the formula
\begin{equation*}
S_T = \exp \left( \sum_{n=1}^\infty (-1)^n \left( \frac{\hbar}{2} \right)^{2n}
    \Bigl( A_n \sgn(x') \abs{2x'}^{-2n+1} \partial_{x'} \partial_{p'}^{2n}
    - B_n \abs{2x'}^{-2n} \partial_{p'}^{2n} \Bigr) \right),
\end{equation*}
where $A_n$ and $B_n$ are rational constants given recursively by
\begin{align*}
A_n & = \frac{1}{2n} \left( 1 - \sum_{k=2}^n \frac{1}{k!} A_{2n-1}^{(k)}
    \right), \\
B_n & = \frac{1}{2n} \left( 2n - 1 - \sum_{k=2}^n \frac{1}{k!} B_{2n-1}^{(k)}
    \right),
\end{align*}
where
\begin{align*}
A_{2n-1}^{(k)} & = \sum_{m=1}^{n-1} 4(n-2m) A_{n-m} A_{2m-1}^{(k-1)}, \\
B_{2n-1}^{(k)} & = \sum_{m=1}^{n-1} \left( 4(n-m) B_{n-m} A_{2m-1}^{(k-1)}
    - 4m A_{n-m} B_{2m-1}^{(k-1)} \right)
\end{align*}
for $k = 2,3,\dotsc,n$ and $n = 2,3,\dotsc$, and
\begin{gather*}
A_1^{(k)} = B_1^{(k)} = 0, \quad k = 2,3,\dotsc, \\
A_{2n-1}^{(1)} = 2n A_n, \quad B_{2n-1}^{(1)} = 2n B_n, \quad n = 1,2,\dotsc.
\end{gather*}
The values of couple first constants $A_n$ and $B_n$ are
\begin{align*}
A_1 & = \frac{1}{2}, & A_2 & = \frac{1}{4}, &
A_3 & = \frac{1}{4}, & A_4 & = \frac{7}{24}, \\
B_1 & = \frac{1}{2}, & B_2 & = \frac{3}{4}, &
B_3 & = \frac{5}{4}, & B_4 & = \frac{49}{24}.
\end{align*}

\subsection{A nonlinear transformation in four-dimensions}
\label{subsec:5.4}
In the following example we will demonstrate the presented theory of coordinate
transformations on the nonlinear transformation $T(x',y',p'_1,p'_2) =
(x,y,p_1,p_2)$ in four-dimensional phase space $\mathbb{R}^4$, where
\begin{align*}
x & = x' - \frac{1}{m_1} t p'_1 - \frac{k}{2m_1} t^2 p'^2_2, \\
y & = y' - \frac{1}{m_2} t p'_2 - 2kt x' p'_2
    + \frac{k}{m_1} t^2 p'_1 p'_2 + \frac{k^2}{3m_1} t^3 p'^3_2, \\
p_1 & = p'_1 + kt p'^2_2, \\
p_2 & = p'_2,
\end{align*}
and $m_1$, $m_2$, $k$, $t$ are some real constants. This transformation is a
classical canonical transformation generated by a function
\begin{equation*}
F(x,y,p'_1,p'_2) = xp'_1 + yp'_2 + ktxp'^2_2 + \frac{1}{2m_1}t p'^2_1
    + \frac{1}{2m_2}t p'^2_2 + \frac{k}{2m_1}t^2 p'_1 p'^2_2
    + \frac{k^2}{6m_1}t^3 p'^4_2.
\end{equation*}
Note, that this transformation is a four-dimensional example of the
transformation from \eqref{eq:3.14a}. It transforms the Moyal \nbr{\star}product
to the following product
\begin{equation*}
f \star'_T g = f \exp \left( \frac{1}{2} i\hbar \overleftarrow{D_{x'}}
    \overrightarrow{D_{p'_1}} + \frac{1}{2} i\hbar \overleftarrow{D_{y'}}
    \overrightarrow{D_{p'_2}} - \frac{1}{2} i\hbar \overleftarrow{D_{p'_1}}
    \overrightarrow{D_{x'}} - \frac{1}{2} i\hbar \overleftarrow{D_{p'_2}}
    \overrightarrow{D_{y'}} \right) g,
\end{equation*}
where
\begin{align*}
D_{x'} & = \partial_{x'} + 2kt p'_2 \partial_{y'}, \\
D_{y'} & = \partial_{y'}, \\
D_{p'_1} & = \partial_{p'_1} + \frac{1}{m_1} t \partial_{x'}
    + \frac{k}{m_1} t^2 p'_2 \partial_{y'}, \\
D_{p'_2} & = \partial_{p'_2} - 2kt p'_2 \partial_{p'_1}
    - \frac{k}{m_1}t^2 p'_2 \partial_{x'} + \left( \frac{1}{m_2}t + 2kt x'
    - \frac{k}{m_1}t^2 p'_1 - \frac{k^2}{m_1}t^3 p'^2_2 \right) \partial_{y'}.
\end{align*}

One can calculate that
\begin{align*}
(\hat{q}'_T)^1 & = x' + \frac{1}{2} i\hbar \partial_{p'_1}
    - \frac{1}{2} \left( \frac{i\hbar}{2} \right)^2 \frac{k}{m_1} t^2
    \partial_{y'}^2, \\
(\hat{q}'_T)^2 & = y' + \frac{1}{2} i\hbar \partial_{p'_2}
    - \left(\frac{i\hbar}{2}\right)^2 \frac{k^2}{m_1} t^3 p'_2 \partial_{y'}^2
    - \left( \frac{i\hbar}{2} \right)^2 \frac{k}{m_1} t^2
    \partial_{x'} \partial_{y'} \\
& \quad {} - \left(\frac{i\hbar}{2}\right)^2 2kt \partial_{y'} \partial_{p'_1}
    + \frac{1}{3} \left( \frac{i\hbar}{2} \right)^3 \frac{k^2}{m_1} t^3
    \partial_{y'}^3, \\
(\hat{p}'_T)_1 & = p'_1 - \frac{1}{2} i\hbar \partial_{x'}
    - \left( \frac{i\hbar}{2} \right)^2 kt \partial_{y'}^2, \\
(\hat{p}'_T)_2 & = p'_2 - \frac{1}{2} i\hbar \partial_{y'},
\end{align*}
and that
\begin{equation*}
[(\hat{q}'_T)^i, (\hat{p}'_T)_j] = i\hbar\delta^i_j.
\end{equation*}
The above equation shows that the transformation $T$ is a quantum canonical
transformation. Now, using \thmref{thm:3.1} it can be shown that the isomorphism
$S_T$ is of the form
\begin{equation*}
S_T = \exp \left(
    \frac{1}{8} \hbar^2 \frac{k}{m_1} t^2 \partial_{x'} \partial_{y'}^2
    + \frac{1}{4} \hbar^2 kt \partial_{p'_1} \partial_{y'}^2
    + \frac{1}{12} \hbar^2 \frac{k^2}{m_1} t^3 p'_2 \partial_{y'}^3 \right).
\end{equation*}
It can be also proved that $S_T$ is an isomorphism (unitary operator) of the
Hilbert space $L^2(\mathbb{R}^4)$ onto itself.

It can be checked using formula \eqref{eq:4.5} that the unitary operator
$\hat{U}_T$ associated to the transformation $T$ is of the form
\begin{align*}
(\hat{U}_T \varphi)(x',y') & = \frac{1}{(2\pi\hbar)^2} \iint \varphi(x,y)
    \sqrt{\Abs{\det\left[ \frac{\partial^2 F}{\partial_{x^i} \partial_{p'_j}}
    (x,y,p'_1,p'_2) \right]}} e^{\frac{i}{\hbar}F(x,y,p'_1,p'_2)}
    e^{-\frac{i}{\hbar}(x'p'_1 + y'p'_2)} \ud{x}\ud{y}\ud{p'_1}\ud{p'_2} \\
& = \frac{1}{(2\pi\hbar)^2} \iint \varphi(x,y)
    e^{\frac{i}{\hbar}((x - x')p'_1 + (y - y')p'_2 + ktxp'^2_2
    + \frac{1}{2m_1}t p'^2_1 + \frac{1}{2m_2}t p'^2_2)} \\
& \quad \times e^{\frac{i}{\hbar}(\frac{k}{2m_1}t^2 p'_1 p'^2_2
    + \frac{k^2}{6m_1}t^3 p'^4_2)} \ud{x}\ud{y}\ud{p'_1}\ud{p'_2}.
\end{align*}

The transformation $T$ is an example of a quantum (and classical) trajectory
on phase space of a two particle system which time evolution is described by the
Hamiltonian \cite{Dias:2004}
\begin{equation*}
H(x,y,p_1,p_2) = \frac{p_1^2}{2m_1} + \frac{p_2^2}{2m_2} + k x p_2^2.
\end{equation*}
The constants $m_1$, $m_2$ are then interpreted as masses of particles, $k$ is
a coupling constant, and $t$ is an evolution parameter (time).

\section{Final remarks}
\label{sec:6}
In the paper we presented consistent theory of canonical transformations of
phase space coordinates in quantum mechanics. The starting point of this
approach was considering phase space quantum mechanics and introducing there
transformations of coordinates in a similar manner as in classical mechanics.
Then we passed to ordinary description of quantum mechanics. Presented theory
allows to solve the problem concerning the consistency of quantization mentioned
in Introduction. Actually, one has to change the ordering of position and
momentum operators in transformed observables from the symmetric to the
\nbr{S_T}ordering.

Returning to the example of the harmonic oscillator from Introduction we can now
associate to the transformed Hamiltonian $H'$ from \eqref{eq:1.2} an appropriate
operator. First, we need to calculate $S_T^{-1}H'$, which using the form of the
isomorphism $S_T$ calculated in the example from \subsecref{subsec:5.3} gives
\begin{equation*}
(S_T^{-1}H')(x',p') = \abs{x'}p'^2 + \omega^2 \abs{x'}
    + \frac{1}{16} \hbar^2 \abs{x'}^{-1}.
\end{equation*}
The \nbr{S_T}ordered operator associated to $H'$ is then equal
\begin{equation*}
H'_{S_T}(\hat{q}',\hat{p}') = (S_T^{-1}H')_M(\hat{q}',\hat{p}')
= \frac{1}{2} \abs{\hat{q}'} \hat{p}'^2 + \frac{1}{2} \hat{p}'^2 \abs{\hat{q}'}
    + \omega^2 \abs{\hat{q}'} + \frac{1}{16} \hbar^2 \abs{\hat{q}'}^{-1}.
\end{equation*}
This operator is indeed unitarily equivalent with its untransformed counterpart.
The unitary operator $\hat{U}_T$ intertwining $H'_{S_T}(\hat{q}',\hat{p}')$ and
$H_M(\hat{q},\hat{p})$ is of the form \eqref{eq:4.4}. To check that this is
correct let us calculate the action of $H'_{S_T}(\hat{q}',\hat{p}')$ on a ground
state of the harmonic oscillator, $\varphi_0(x) = \left(\frac{\omega}{\pi\hbar}
\right)^{1/4} \exp\left(-\frac{\omega x^2}{2\hbar}\right)$, transformed to the
new coordinate system by $\hat{U}_T$. One easily calculates that
\begin{equation*}
(\hat{U}_T \varphi_0)(x') = \left(\frac{\omega}{\pi\hbar}\right)^{1/4}
    \frac{1}{\sqrt{\abs{2x'}}} \exp\left(-\frac{\omega \abs{x'}}{\hbar}\right)
\end{equation*}
and that
\begin{equation*}
H'_{S_T}(\hat{q}',\hat{p}') \hat{U}_T \varphi_0 = \frac{1}{2}\hbar\omega
    \hat{U}_T \varphi_0,
\end{equation*}
which shows that $\hat{U}_T \varphi_0$ is an eigen-state of the transformed
Hamiltonian of the oscillator, corresponding to an energy
$\frac{1}{2}\hbar\omega$, as it should be.

The general theory of quantum canonical transformations developed in our paper
suggests the proper approach for tackling the problem of general transformations
of coordinates in quantum mechanics. Especially it provides a method of an
investigation of non-canonical transformations, which are hard to introduce
directly in ordinary quantum mechanics. In fact, the proper way of developing
non-canonical transformation theory should look as follows. First we consider a
transformation $T$ of coordinates in phase space quantum mechanics and find the
form of the isomorphism $S_T$ related to $T$. Then with the help of $S_T$ we
can pass to ordinary quantum mechanics. The problem which still needs to be
solved is an introduction of a proper ordering of operators related to
non-canonical coordinates $(\xi,\zeta)$, i.e., the definition of the operator
function has to be extended so that the following formula would hold
\begin{equation*}
A \star'_T {} = A_{S_T}(\xi \star'_T {}, \zeta \star'_T {}).
\end{equation*}
Only then the passage to ordinary quantum mechanics can be found.

In this paper we considered transformations of coordinates which were defined
on almost the whole phase space. Problems arise when a transformation will be
defined on some arbitrary open subset of phase space. We believe that in such
case the space of states in ordinary quantum mechanics will be a Hilbert space
of square integrable functions with respect to some measure dependent on the
transformation. Moreover, the momentum operator would not be in the form
$-i\hbar\partial_{x'}$.

From the application point of view the problem which has to be solved is
a systematic construction of isomorphisms $S_T$ and operators
$\hat{U}_T$ for a general transformation $T$.

In conclusion, the possibility of introducing different coordinate systems in
quantum mechanics is promising, however, there are still many problems which
need to be solved. Furthermore, the geometric approach to coordinate
transformations in phase space quantum mechanics, especially for quantum
trajectories, gives a potential possibility of developing quantum theories
similar to those found in classical theory, e.g., to develop a geometric theory
of quantum integrable systems: quantum bi-Hamiltonian systems.


\end{document}